%% file: CSP.tex
\documentclass[a4paper,12pt,eqno]{article}
\usepackage{latexsym,amssymb,amsfonts,amsmath,amsthm}
\usepackage{graphicx}
\newtheorem{theorem}{Theorem}[section]
\newtheorem{lemma}{Lemma}[section]

\newtheorem{proposition}{Proposition}[section]
\newtheorem{remark}{Remark}[section]
\newtheorem{definition}{Definition}[section]

\usepackage{ulem}
\usepackage{comment}

\usepackage{color}
\usepackage{enumerate}

\newcommand{\tr}{\mathrm{tr}}

\title{{\Large {\bf A crossover between open quantum random walks to quantum walks}}
\author{
{\small Norio Konno}\\
{\scriptsize Department of Applied Mathematics, Faculty of Engineering, Yokohama National University, }\\ {\scriptsize Yokohama, 240-8501, Japan}\\
{\small Kaname Matsue}\\
{\scriptsize Institute of Mathematics for Industry, Kyushu University, }\\
{\scriptsize Fukuoka 819-0395, Japan}\\
{\scriptsize International Institute for Carbon-Neutral Energy Research, Kyushu University, }\\
{\scriptsize Fukuoka 819-0395, Japan}\\
{\scriptsize Center for Research and Development Strategy, Japan Science and Technology Agency (JST-CRDS),}\\
{\scriptsize Tokyo 102-0076, Japan,}\\
{\small Etsuo Segawa\footnote{segawa-etsuo-tb@ynu.ac.jp}} \\
{\scriptsize Graduate School of Environment and Information Sciences, Yokohama National University} \\
{\scriptsize Yokohama, 240-8501, Japan}}
}
\date{\empty }
\begin{document}
\maketitle

\par\noindent
\begin{small}
\par\noindent
{\bf Abstract} We propose an intermediate walk continuously connecting an open quantum random walk and a quantum walk with parameters $M\in\mathbb{N}$ controlling a decoherence effect; if $M=1$, the walk coincides with an open quantum random walk, while $M=\infty$, the walk coincides with a quantum walk. We define a measure which recovers usual probability measures on $\mathbb{Z}$ for $M=\infty$ and $M=1$ and we observe intermediate behavior through  numerical simulations for varied positive values $M$. In the case for $M=2$, we analytically show that a typical behavior of quantum walks appears even in a small gap of the parameter from the open quantum random walk. More precisely, we observe both the ballistically  moving towards left and right sides and localization of this walker simultaneously. The analysis is based on Kato's perturbation theory for linear operator. We futher analyze this limit theorem in more detail and show that the above three modes are described by Gaussian distributions.
\section{Introduction}
Quantum walks with small perturbations exhibit ballistic spreading and localization, simultaneously~\cite{Konno2008b}. The existence of the wave operators implies that the limit distributions are described by a linear combination of a continuous function having a finite support and a delta measure at the origin in the super diffusive scaling with respect to the time step~\cite{Suzuki}. Such properties are supposing the evidences of the quantum speed up algorithms driven by quantum walks on graphs~\cite{Por} because quantum walks spread quadratically faster than random walks and inform us the perturbation's place by localization. 

Quantum walks are expected to simulate quantum dynamics, for example, in the topological phase, Dirac equations, and so on. Finding a framework for quantum walks in an open system have been also needed in particular by quantum biology. Attal et al. proposed the open quantum random walk. Detailed limit theorems for the open quantum random walks are developed by for example~\cite{AttalEtAl1, DM, KKSY} in terms of probability theory and quantum probability theory. The limit theorem is so-called the central limit theorem which implies that the scaling order is diffusive and the limit shape of the distribution follows the Gaussian distribution. 

In this paper, we attempt to interpolate such a gap in the limit theorems between quantum walks and open quantum random walks.  We restrict ourselves to a space homogeneous quantum walk on one dimensional lattice $\mathbb{Z}$ and its induced open quantum random walk. 
Here this open quantum walk satisfies the restrictive condition (10) in \cite{AttalEtAl2}. 
The key idea for it is considering a quantum walk on two dimensional lattice $\mathbb{Z}^2$ instead of considering $\mathbb{Z}$ directly. 
We set a stripe which is parallel with the diagonal line of $\mathbb{Z}^2$; $D_{s,t}\subset \mathbb{Z}^2$ by $\{ (x,y)\in \mathbb{Z}^2  \;|\; s\leq x-y \leq t\}$ for some parameters $s\leq 0 \leq t$. We introduce a four-state-quantum walk on $\mathbb{Z}^2$ with the Dirichret-cut of this stripe $D_{s,t}$.  Remark that this time evolution is no longer unitary and also loses some regularities. 

However, we can reproduce the probability distributions of the open quantum random walk and the unitary quantum walk by introducing a measure on the diagonal line if we set the parameters by $t=s=0$ and $t=\infty,\;s=-\infty$, respectively.  Then the width of the stripe controls the decoherence of the walk because the coherent factors remain in the off diagonal ranges of $D_{s,t}$.  
Indeed, our model corresponds to the walk extending the decoherence step 2) in the realization procedure of the open quantum random walk in Proposition~8.1 \cite{AttalEtAl2}
by newly introducing parameters $s,t\in\mathbb{R}$ with $s<t$. 
Remark that if $s=t=0$, then the original decoherence step is recovered, while if $s=-\infty$ and $t=\infty$ 
then because the situation is equivalent to skipping the decoherence step, a unitary quantum walk is recovered (see Proposition~10.1~\cite{AttalEtAl2}). 
We discuss this relation in Section~7 in more detail. 
The induced quantum random walk treated here is unitarily equivalent to a correlated random walk~\cite{Konno_Cor}. Then our model also connects quantum walks to random walks. It is therefore worthy to note some related works on quantum walks, 
for example, quantum walks with decoherence \cite{Brun,Kendon2007}, 
crossover from random walk to quantum walk behavior with multiple coins~\cite{SeKo}. The continuous time version can be seen in the stochastic quantum walk introduced by~\cite{WRA} dividing the Kossakowski-Lindblad equation into a convex combinations of coherence part and stochastic part with one parameter.  
In this paper, we demonstrate that the width of stripe controls the decoherence effect by a numerical simulation, and mathematically provide the spectral analysis based on the Kato perturbation theory for linear operator~\cite{Kato1982}  (Theorem~\ref{thm:eigenprojection}) and limit of this measure for $|s-t|=1$ (Theorem~\ref{theorem-limit-behavior}).

This paper is organized as follows.
In Section~2, we give the definition of our proposal model and introduce a measure on the one dimensional lattice. 
In Section~3, we consider the Fourier transform of this walk and obtain the limit theorems for the induced open quantum random walk which may be a repetition of the previous works to see an idea of the analysis on the general case. 
In Section~4, we consider the spectral analysis on the walk especially for $|s-t|=1$ case using Kato's perturbation theories~\cite{Kato1982}. 
In Section~5, we devote to the limit theorems of the measure for the case of  $|s-t|=1$ from the results on Section~4. 
In Section~6, the behaviors of the measure by the numerical simulations for varied parameters $M:=|s-t|+1$. Finally, we give a summary and discussion in Section~7.
\section{Model} 
\subsection{Quantum walk and induced open quantum random walk}
The Hilbert space of the quantum walk treated here is denoted by $\ell^2(\mathbb{Z};\mathbb{C}^2)$. 
The time evolution operator of the quantum walk treated here is defined by 
    \begin{equation}\label{eq:TQW}
        (U\psi)(x) = P'\psi(x+1)+Q'\psi(x-1)
    \end{equation}
for any $\psi\in \ell^2(\mathbb{Z};\mathbb{C}^2)$. 
Here $P'$ and $Q'$ are $2$-dimensional matrices defined by 
    \[ P'=\begin{bmatrix}a & b \\ 0 & 0\end{bmatrix},\;
    Q'=\begin{bmatrix}0 & 0 \\c & d\end{bmatrix}, \]
where $H:=P'+Q'$ is the $2$-dimensional unitary matrix with $abcd\neq 0$. 
The canonical basis of $\mathbb{C}^2$ is denoted by $|L\rangle=[1\;0]^\top$ and $|R\rangle=[0\;1]^\top$, respectively. Then it holds $P'=|L\langle\rangle L| H$ and $Q'=|R\rangle \langle R| H$.
Let $\psi_n^{(1)}$ be the $n$-th iteration of quantum walk time evolution such that  $\psi_{n+1}^{(1)}=U\psi_n^{(1)}$. 
Then the distribution at each time $n$; $\mu_n: \mathbb{Z}\to [0,1]$, can be described by 
    \begin{equation}\label{eq:mun} 
    \mu_n(x)=||\psi_n^{(1)}(x)||^2. 
    \end{equation}
Let us introduce an equivalent expression of the time evolution of quantum walk (\ref{eq:TQW}) as follows. 
Let $\psi_{n}^{(2)}\in \ell^2(\mathbb{Z}^2;\mathbb{C}^4)$ obeys the following recursion such that 
    \begin{multline}\label{eq:2dTQW}
        \psi_{n+1}^{(2)}(x,y) = 
        (P\otimes \bar{P})\psi_{n}^{(2)}(x+1,y+1)
       +(Q\otimes \bar{P})\psi_{n}^{(2)}(x-1,y+1) \\
       +(P\otimes \bar{Q})\psi_{n}^{(2)}(x+1,y-1)
       +(Q\otimes \bar{Q})\psi_{n}^{(2)}(x-1,y-1),
    \end{multline}
where $P=H|L\rangle \langle L|$ and $Q=H|R\rangle \langle R|$ and $\bar{P}$ and $\bar{Q}$ are the complex conjugates of $P$ and $Q$, respectively.
Here we describe the canonical basis of $\mathbb{C}^4$ by $|LL\rangle:=[1\;0\;0\;0]^\top$, 
$|LR\rangle:=[0\;1\;0\;0]^\top$, $|RL\rangle:=[0\;0\;1\;0]^\top$, $|RR\rangle:=[0\;0\;0\;1]^\top$. 
Then the distribution of quantum walk is expressed as follows: 
\begin{proposition} Let $\psi_n^{(1)}$, $\psi_n^{(2)}$ and $\mu_n$ be the above. Assume the initial state of $\psi_n^{(1)}$ is $\psi_0^{(1)}(x)=\delta_0(x)\varphi_0$ with some unit vector  $\varphi_0\in \mathbb{C}^2$.  
If the initial state of $\psi_n^{(2)}$ is $\psi_0^{(2)}(x,y)=\delta_{(0,0)}(x,y)(H\varphi_0)\otimes \overline {(H\varphi_0)}$, then we have  
\[ \mu_n(x)=\langle LL| \psi_n^{(2)}(x,x)\rangle + \langle RR|\psi_n^{(2)}(x,x)\rangle  \]
for any $n\geq 0$ and $x\in \mathbb{Z}$. 
\end{proposition}
\begin{proof}
By the definition of the time iteration of the 
$1$-dimensional quantum walk, it is easy to see that $\psi_n^{(1)}(x)\psi_n^{(1)}(y)^*=:\Psi_n(x,y)\in M_2(\mathbb{C})$  satisfies 
   \begin{multline}\label{eq:2DQW}
        \Psi_{n+1}(x,y) = 
        P'\Psi_n(x+1,y+1){P'}^*
       +Q'\Psi_n(x-1,y+1){P'}^* \\
       +P'\Psi_n(x+1,y-1){Q'}^*
       +Q'\Psi_n(x-1,y-1){Q'}^*.
    \end{multline}
Then the probability distribution at time $n$ is described by   
    \begin{align*} 
    \mu_{n+1}(x) 
    &= \mathrm{tr}(\Psi_{n+1}(x,x)) \\
    &= \mathrm{tr}({P'}^*P'\Psi_n(x+1,x+1))+\mathrm{tr}({Q'}^*Q'\Psi_n(x-1,x-1) \\
    & \qquad\qquad + \mathrm{tr}({P'}^*Q'\Psi_n(x+1,x-1))+\mathrm{tr}({Q'}^*P'\Psi_n(x-1,x+1))\\
    &= \mathrm{tr} ( |-\rangle \langle -| \Psi_n(x+1,x+1))
      +\mathrm{tr}(|+\rangle \langle +|\Psi_n(x-1,x-1)).
    \end{align*}
Here we used ${P'}^*Q'={Q'}^*P'=0$ and we put $|- \rangle:=H^*|L\rangle$ and $|+\rangle:=H^*|R\rangle$. 
Inspired by this expression of the distribution, we consider $v_n\in (\mathbb{C}^4)^{\mathbb{Z}^2}$ which is isomorphic to $\Psi_n$ after the map $ (M_2(\mathbb{C}^2))^{\mathbb{Z}^2}\to (\mathbb{C}^4)^{\mathbb{Z}^2}$ such that  
for $x,y\in \mathbb{Z}$   
\begin{equation}  
v_n(x,y) := [ \langle -|\Psi_{n}(x,y)|-\rangle  \;\;\langle -|\Psi_{n}(x,y)|+\rangle\;\;\langle +|\Psi_{n}(x,y)|-\rangle\;\;\langle +|\Psi_{n}(x,y)|+\rangle]^\top.
\end{equation}
Let us see $v_n$ coincides with $\psi_n^{(2)}$ in the following. 
Note that the distribution $\mu_n(x)$ is represented by 
    \[ \mu_n(x) = \tr(\Psi_n(x,x)) = \tr((|-\rangle\langle-|+|+\rangle\langle +|)\Psi_n(x,x))
    = \langle LL|v_n(x,x)\rangle + \langle RR|v_n(x,x)\rangle. \]
Multiplying $|\epsilon \rangle \langle \epsilon'|$ to both sides of (\ref{eq:2DQW}) and taking trace for each case  $(\epsilon,\epsilon' \in \{\pm \})$, 
we obtain 
    \begin{align*} 
    v_{n+1}(x,y) 
    &= 
    \begin{bmatrix}  
    |a|^2 & 0 & 0 & 0 \\ a\bar{c} & 0 & 0 & 0 \\ c\bar{a} & 0 & 0 & 0 \\ |c|^2 & 0 & 0 & 0 
    \end{bmatrix} v_n(x+1,y+1)  
+   \begin{bmatrix}  
    0 & a\bar{b} & 0 & 0 \\ 0 & a\bar{d} & 0 & 0 \\ 0 & c\bar{b} & 0 & 0 \\ 0 & c\bar{d} & 0 & 0 
    \end{bmatrix} v_n(x+1,y-1)  \\
&\qquad\qquad +   \begin{bmatrix}  
    0 & 0 & b\bar{a} & 0 \\ 0 & 0 & b\bar{c} & 0 \\ 0 & 0 & d\bar{a} & 0 \\ 0 & 0 & d\bar{c} & 0 
    \end{bmatrix} v_n(x-1,y+1)  
+   \begin{bmatrix}  
    0 & 0 & 0 & |b|^2 \\ 0 & 0 & 0 & b\bar{d} \\ 0 & 0 & 0 & d\bar{b} \\ 0 & 0 & 0 & |d|^2 
    \end{bmatrix} v_n(x-1,y-1)  \\
    \\
    &= (P\otimes \bar{P}) v_n(x+1,y+1) + (P\otimes \bar{Q})  v_n(x+1,y-1) \\
    &\qquad\qquad + (Q\otimes \bar{P}) v_n(x-1,y+1)+ (Q\otimes \bar{Q}) v_n(x+1,y+1).
    \end{align*}
Therefore $v_n$ satisfies the same recursion of $\psi_n^{(2)}$. 
The matrix representation of the initial state for $\psi_n^{(1)}$ is expressed by $\Psi_0(x,y)=\delta_{(0,0)}(x,y)\varphi_0 \varphi_0^*$. 
Then the corresponding initial state  $v_0(x,y)$ must be 
\begin{align*}
    v_0(x,y) 
    &= 
    \delta_{(0,0)}(x,y)[\langle -|\varphi_0\varphi_0^*|-\rangle \rangle\;
     \langle -|\varphi_0\varphi_0^*|+\rangle \rangle\;
     \langle +|\varphi_0\varphi_0^*|-\rangle \rangle\;
     \langle +|\varphi_0\varphi_0^*|+\rangle \rangle\;]^\top \\
    &= 
    \delta_{(0,0)}(x,y)
    \begin{bmatrix}
    \langle L|H|\varphi_0\rangle \\
    \langle R|H|\varphi_0\rangle
    \end{bmatrix} 
    \otimes 
    \begin{bmatrix}
    \langle \varphi_0|H^*|L\rangle \\
    \langle \varphi_0|H^*|R\rangle
    \end{bmatrix} \\
    &= 
    \delta_{(0,0)}(x,y)
    H\varphi_0 \otimes \overline{H\varphi_0}
\end{align*}
which implies $v_n=\psi_n^{(2)}$ and completes the proof. 
\end{proof}

On the other hand, the corresponding time evolution operator of the open quantum random walk treated here is 
\begin{equation}\label{eq:TOQRW}
        (\mathcal{M}\Phi)(x) = P'\Phi(x+1){P'}^*+Q'\Phi(x-1){Q'}^*
\end{equation}
for every density matrix $\Phi=\sum_{j}p_j \psi_j\psi_j^*$ with $||\psi_j||_{\ell^2(\mathbb{Z};\mathbb{C}^2)}=1$ for any $j$ and $\sum_{j}p_j=1$. 
Let $\Phi_n$ be the $n$-th iteration of (\ref{eq:TOQRW}) such that $\Phi_{n+1}=\mathcal{M}\Phi_n$. 
By the trace preserving property ${P'}^*P'+{Q'}^*Q'=I$, the distribution at each time step is defined by $m_n(x)=\tr(\Phi_n(x))$. 
Then in the same way as the quantum walk case, we obtain the following proposition.
\begin{proposition}
Let $p_n(x)\in \mathbb{C}^2$ be $[ \langle -|\Phi_n(x)|-\rangle\;\; \langle +|\Phi_n(x)|+\rangle]^\top$. Then $p_n$ satisfies the following recursion.   
\begin{equation}\label{eq:OQRWp}
p_{n+1}(x) = (Q\circ \bar{Q}) p_n(x-1) +  (P\circ \bar{P}) p_n(x+1).  \end{equation}
Here $\circ$ is the Hadamard product of matrices. 
The distribution is described by 
\begin{equation}\label{eq:mn} m_n(x)= \langle L|p_n(x)\rangle + \langle R|p_n(x)\rangle.\end{equation} 
\end{proposition}
\subsection{Intermediate walk between quantum walk and open quantum random walk: Dirichlet-cut quantum walk model}
In the next, we connect continuously the quantum walk and the induced open quantum random walk using some parameters. 
To this end,  
let us restrict the domain of the time evolution of the quantum walk on $2$-dimensional lattice (\ref{eq:2dTQW}) by considering the Dirichlet boundary condition: 
Let $D_{s,t} :=\{ (x,y)\in\mathbb{Z}^2 \;|\; s\leq x-y \leq t \}$ for $s\leq 0\leq  t$ and 
\begin{multline}\label{eq:U2} 
(U^{(2)}\psi)(x,y)
=(P \otimes \bar{P})\psi(x+1,y+1)
       +(Q \otimes \bar{P})\psi(x-1,y+1) \\
       +(P \otimes \bar{Q})\psi(x+1,y-1)
       +(Q \otimes \bar{Q})\psi(x-1,y-1), 
\end{multline}
for any $\psi\in \ell^2(\mathbb{Z}^2;\mathbb{C}^4)$. 
Then we define the Dirichret-cut quantum walk $\{\phi_n^{(s,t)}\}_n$ as follows: 
\begin{definition}
Let $s\leq 0\leq t$. 
If sequence of $\{\phi_n^{(s,t)}\}_n$ satisfies the following time evolution, we call this walk  
the Dirichret-cut quantum walk  with respect to the boundary $D_{s,t}\subset \mathbb{Z}^2$. 
\begin{equation}
\phi_{n+1}^{(s,t)}(x,y)= 
\begin{cases}
(U^{(2)}\phi_n^{(s,t)})(x,y) & \text{: $(x,y)\in D_{s,t}$,}\\ 
0, & \text{: otherwise,}        
\end{cases}    
\end{equation}
\[ \phi_0^{(s,t)}(x,y) = \delta_{(0,0)}(x,y)(H\varphi_0 \otimes \overline{H\varphi_0}).\]
Here $\varphi_0$ is a unit vector on $\mathbb{C}^2$. 
\end{definition}
We have a simple but important observation of the Dirichret-cut quantum walk as follows. 
Recall that the probability distributions $\mu_n$ for the quantum walk and $m_n$ for the open quantum random walk are defined in (\ref{eq:mun}) and (\ref{eq:mn}), respectively.
\begin{proposition}\label{prop:0infty}
Let $\mu_n$ and $m_n$ be the above. Then we have 
\begin{align} 
\mu_n(x) &= \langle LL| \phi_n^{(-\infty,\infty)}(x,x)\rangle+\langle RR| \phi_n^{(-\infty,\infty)}(x,x)\rangle, \label{eq:mun2}\\
m_n(x) &= \langle LL| \phi_n^{(0,0)}(x,x)\rangle+\langle RR| \phi_n^{(0,0)}(x,x)\rangle.
\end{align}
\end{proposition}
\begin{proof}
The time evolution of the case for $(t,s)=(-\infty,\infty)$ coincides with  (\ref{eq:2dTQW}). Then (\ref{eq:mun2}) holds. 
On the other hand, the time evolution of the case for $(t,s)=(0,0)$ is expressed by
\[ \phi_{n+1}^{(0,0)}(x,x) = (Q\otimes \bar{Q}) \phi_{n}^{(0,0)}(x-1,x-1)+ (P\otimes \bar{P}) \phi_{n}^{(0,0)}(x+1,x+1). \]
Then 
\begin{align*}
 \langle LL|\phi_{n+1}^{(0,0)}(x,x)\rangle
    &= |a|^2 \langle LL|\phi_{n}^{(0,0)}(x+1,x+1) \rangle + |b|^2 \langle RR|\phi_{n}^{(0,0)}(x-1,x-1) \rangle,   \\
\langle RR|\phi_{n+1}^{(0,0)}(x,x)\rangle
    &= |c|^2 \langle LL|\phi_{n}^{(0,0)}(x+1,x+1) \rangle+ |d|^2 \langle RR|\phi_{n}^{(0,0)}(x-1,x-1) \rangle,   
\end{align*}
which implies that $[ \langle LL|\phi_{n}^{(0,0)}(x,x)\rangle \;\;\langle RR|\phi_{n}^{(0,0)}(x,x)\rangle]^\top$ satisfies the recursion of the open quantum random walk (\ref{eq:OQRWp}).   
Therefore we obtain the conclusion.
\end{proof}
We are interested in the case for $(s,t) \notin \{(0,0),(-\infty,\infty)\}$ since it is an intermediate region between the random walk and the quantum walk. 
To see the crossover, we focus on the ``measure" which may takes complex value in general except $(s,t)=(0,0)$ and $(s,t)=(-\infty,\infty)$.
\begin{definition}\label{def:measure}
The complex valued measure $\mu_n^{(s,t)}: \mathbb{Z}\to \mathbb{C}$ is defined by 
\[ \mu_n^{(s,t)}(x):= \langle LL| \phi_n^{(s,t)}(x,x) \rangle + \langle RR| \phi_n^{(s,t)}(x,x) \rangle. \]
\end{definition}
Remark that $\mu_n^{(0,0)}=m_n$ and $\mu_n^{(-\infty,\infty)}=\mu_n$.
Then the following limit theorems for $\mu_n^{(0,0)}$ and $\mu_n^{(-\infty,\infty)}$ are obtained: 
letting the initial state be $\delta_{0,0}(x,y)(|LL\rangle+|RR\rangle)/2$, which corresponds to the mixed state at the origin with respect to the internal state, then it holds that for any $\tau\in \mathbb{R}$, 
\begin{align} 
    \lim_{n\to\infty}\sum_{x\leq \sqrt{n}\tau}\mu_n^{(0,0)}(x)
    &= \int_{-\infty}^{\tau} \frac{e^{-s^2/(2\sigma^2)}}{\sqrt{2\pi \sigma^2}}ds, \label{eq:OQRWcase} \\ 
    \lim_{n\to\infty}\sum_{x\leq n\tau}\mu_n^{(-\infty,\infty)}(x) 
    &= \int_{-\infty}^{\tau} \frac{|b|}{\pi (1-s^2)\sqrt{|a|^2-s^2}}\boldsymbol{1}_{\{|s|<1/\sqrt{2}\}}(s)ds, \label{eq:QWcase}
\end{align}
where $\sigma^2= |a|^2/|b|^2$. 
Note that the scaling orders are $\sqrt{n}$ (diffusive) and $n$ (ballistic), respectively and the shapes are the Gaussian and Konno distributions, respectively. 
The proofs for (\ref{eq:OQRWcase}) and (\ref{eq:QWcase}) can be seen for example, \cite{Konno2008b} and its references therein. 
In the next section, we revisit the formula (\ref{eq:OQRWcase}) to see a fundamental idea of the analysis for $(s,t)\in \{(0,1),(1,0)\}$.  
\section{Fourier transform}
\subsection{Fourier transform}
Assume that the initial state of $\phi_n^{(s,t)}$ is expressed by $\phi_0^{(s,t)}(x,y)=\varphi_0' \otimes \overline{\varphi_0'}$ with $||\varphi_0'||^2=1$.  
Let $\hat{\mu}_n^{(s,t)}(k)$ be the Fourier transform of $\mu_n^{(s,t)}(x)$ such that $\hat{\mu}_n^{(s,t)}(k):=\sum_{x\in \mathbb{Z}}\mu_n^{(s,t)}(x)e^{ikx}$ for $k\in \mathbb{R}/2\pi\mathbb{Z}$. 
Remark that $\hat{\mu}_n^{(0,0)}$ and $\hat{\mu}_n^{(-\infty,\infty)}$ coincide with the characteristic functions of the distributions at time $n$ for the open quantum random walk and quantum walk on $\mathbb{Z}$, respectively. 
Therefore (\ref{eq:OQRWcase}) and (\ref{eq:QWcase}) imply 
\begin{align*}
    \lim_{n\to\infty}\hat{\mu}_n^{(0,0)}(k/\sqrt{n}) 
        &= \int_{-\infty}^{\infty} 
        \frac{e^{-s^2/(2\sigma^2)}}{\sqrt{2\pi \sigma^2}} e^{-isk}ds = 
        e^{-k^2/(2\sigma^2)}, \\
    \lim_{n\to\infty}\hat{\mu}_n^{(-\infty,\infty)}(k/n), 
        &= \int_{-\infty}^{\infty} \frac{|b|e^{-iks}}{\pi (1-s^2)\sqrt{|a|^2-s^2}}\boldsymbol{1}_{\{|s|<1/\sqrt{2}\}}(s)ds.
\end{align*}
Our purpose is to find an appropriate parameter $\theta$ so that $\hat{\mu}_n^{(s,t)}(k/n^\theta)$ converges.  If $\theta=1/2$, we call the diffusive scaling and if $\theta=1$, we call the ballistic scaling. 

Since the Dirichret-cut quantum walk $\phi_n^{(s,t)}$ moves diagonally on $\mathbb{Z}^2$, 
we take the $-\pi/4$ rotation to $\mathbb{Z}^2$ and reduce the scale of  $1/\sqrt{2}$ so that each step length is normalized and the diagonal line $y=x$ is rotated to the horizontal axis; that is, 
$(x,y)\mapsto (u,v)$, where $u=(x+y)/2$, $v=(x-y)/2$. 
Because of the parity of this walk, we can assume the new coordinates $u$ and $v$ are integers at least the walk starts from the origin. 
Then letting $\mathbb{Z}_e^2:=\{(u+v,u-v) \;|\; u,v\in \mathbb{Z}\}$, we define the Fourier transform 
$\ell^2(\mathbb{Z}_e^2;\mathbb{C}^4)\to L^2([0,2\pi)\times \mathbb{Z};\mathbb{C}^4)$
by 
\[ (\mathcal{F}\phi)(k;v)=\sum_{u\in \mathbb{Z}} \phi(u+v,u-v)e^{iku} \;\;(v\in \mathbb{Z}).
    \]
Its inverse Fourier transform is described by 
    \[ (\mathcal{F}^{*}\hat{\phi})(u+v,u-v)=\int_{0}^{2\pi} \hat{\phi}(k;v)e^{-iku} \frac{dk}{2\pi}. \]
We put $(\mathcal{F}\phi_n^{(s,t)})(k;v):=\hat{\phi}_n^{(s,t)}(k;v)$.
Note that by the parity of this walk, we have $U^{(2)}(\ell^2(\mathbb{Z}^2_e;\mathbb{C}^4))=\ell^2(\mathbb{Z}^2_e;\mathbb{C}^4)$. 
\begin{remark}
It holds that 
    \[\langle LL|\hat{\phi}_n^{(s,t)}(k;0)\rangle + \langle RR|\hat{\phi}_n^{(s,t)}(k;0)\rangle =\hat{\mu}_n^{(s,t)}(k).\] 
\end{remark}
For fixed $k\in [0,2\pi)$, let $\hat{U}(k)$ be the following unitary operator on $\ell^2(\mathbb{Z};\mathbb{C}^4)$ defined by  
    \[ (\hat{U}(k)\hat{\phi})(v) =  (P\otimes \bar{Q})\hat{\phi}(v+1)+ (Q\otimes \bar{P})\hat{\phi}(v-1)+V(k)\hat{\phi}(v), \]
which is identical with $\mathcal{F}U^{(2)}\mathcal{F}^*$, 
where $V(k)=e^{-ik}(P\otimes \bar{P})+e^{ik}(Q\otimes \bar{Q})$. 
Then we have 
\begin{equation}
\label{eq:FT} \hat{\phi}_{n+1}^{(s,t)}(k;v) 
= \begin{cases}
(\hat{U}(k)\hat{\phi}_{n}^{(s,t)}(k;\cdot))(v) & \text{: $v\in \{s,\dots,t\}$, } \\ 0 & \text{: otherwise. }
\end{cases} 
\end{equation}
Let $\chi_{s,t}: \ell^2(\mathbb{Z};\mathbb{C}^4)\to \ell^2(\{s,\dots,t\};\mathbb{C}^4)$ be the boundary operator such that $(\chi_{s,t}\hat{\phi})(x)=\hat{\phi}(x)\;(x\in \{s,\dots,t\})$. Then its adjoint $\chi_{s,t}^*: \ell^2(\{s,\dots,t\};\mathbb{C}^4)\to \ell^2(\mathbb{Z};\mathbb{C}^4)$ is  
\[ (\chi_{s,t}^*\varphi)(x)=\begin{cases} \varphi(x) & \text{: $x\in \{s,\dots, t\}$,} \\ 0 & \text{: otherwise.} \end{cases} \]
Putting $\hat{W}_{s,t}(k):=\chi_{s,t} \hat{U}(k) \chi^*_{s,t}$. By (\ref{eq:FT}), 
we have 
    \begin{align*} 
    \chi_{s,t} \hat{\phi}_{n+1}^{(s,t)}(k;\cdot) &= \chi_{s,t} \hat{U}(k) (\chi_{s,t}^* \chi_{s,t} + (1-\chi_{s,t}^*\chi_{s,t}))\hat{\phi}_n^{(s,t)}(k;\cdot)\\
     &= \hat{W}_{s,t}(k)\chi_{s,t}\hat{\phi}_n^{(s,t)}(k;\cdot) + \chi_{s,t} \hat{U}(k)(1-\chi_{s,t}^*\chi_{s,t})\hat{\phi}_n^{(s,t)}(k;\cdot) \\
     &= \hat{W}_{s,t}(k)\chi_{s,t}\hat{\phi}_n^{(s,t)}(k;\cdot)
    \end{align*}
for any $v\in\{s,\dots,t\}$. 
Then the problem is reduced to the spectral analysis on the finite matrix of $\hat{W}_{s,t}(k)$ because the characteristic function $\mu_n^{(s,t)}(x)$ is expressed by
    \begin{align*}
        \hat{\mu}_n^{(s,t)}(k) 
        &= \langle LL|\hat{\phi}_n^{(s,t)}(k;0)\rangle + \langle RR|\hat{\phi}_n^{(s,t)}(k;0)\rangle \\
        &= \left\langle q_0 , \hat{W}^n_{s,t}(k) \hat{\phi}^{(s,t)}_0(k) \right\rangle, 
    \end{align*} 
where $q_0(x):=\delta_0(x)(|LL\rangle +|RR\rangle)$. 
Here the initial state in the Fourier space is 
$\hat{\phi}_0^{(s,t)}(k;v)=\delta_{0}(v) \varphi_0'\otimes \overline{\varphi_0'}$. 
It holds $\hat{W}_{s,t}=\chi_{s,t}\mathcal{F} U^{(2)}\mathcal{F}^*\chi_{s,t}^*$. 
\subsection{Limit theorem for $(s,t)=(0,0)$ case (Induced open quantum random walk)}
Assume the initial state of $\phi_{n}^{(0,0)}$ be $\phi_0^{(0,0)}(x,y)=\delta_{(0,0)}(x,y)\varphi_0'\otimes \overline{\varphi_0'}$. 
By Proposition~\ref{prop:0infty}, 
putting $p_n(x):=[\langle LL|\phi_n^{(0,0)}(x,x)\rangle\;\; \langle RR|\phi_n^{(0,0)}(x,x)\rangle]^\top$, we obtain
\begin{equation}\label{eq:corRWrec} 
p_{n+1}(x)=
\begin{bmatrix} |a|^2 & 0 \\ |c|^2 & 0
\end{bmatrix} p_n(x-1)+\begin{bmatrix} 0 & |b|^2 \\ 0 & |d|^2
\end{bmatrix}p_n(x+1).
\end{equation}
By taking the Fourier transform; $\hat{p}_n(k):=\sum_{x\in \mathbb{Z}}p_n(x)e^{ikx}$, 
we obtain 
    \[ \hat{p}_{n+1}(k)=V'(k)\hat{p}_n(k).  \]
Here we put $r:=|a|^2$, and 
    \[ V'(k):=\begin{bmatrix} e^{ik}r & e^{-ik}(1-r) \\ e^{ik} (1-r) & e^{-ik}r  \end{bmatrix}. \]
The eigenvalues $\lambda_{\pm}(k)$ of $V'(k)$ are the solutions of the following quadratic equation
    \[ \lambda^2-2r\cos k \lambda+(2r-1)=0. \]
More precisely, the solutions are 
$\lambda_\pm(k)=r\cos k\pm \sqrt{r^2\cos^2 k-(2r-1)}$. 
Put $\varphi_0'=[ \alpha\; \beta ]^\top$ with $|\alpha|^2+|\beta|^2=1$. 
The characteristic function is expressed by 
    \begin{align}  
    \hat{\mu}_n^{(0,0)}(k)
    &=(\langle L|+\langle R|)\hat{p}_n(k)
    = (\langle L|+\langle R|)(|\alpha|^2{V'}^n(k)|L\rangle + |\beta|^2{V'}^n|R\rangle) \notag \\
    &= |\alpha|^2 ({V'}^n(k))_{LL}+|\beta|^2({V'}^n(k))_{RR}+|\alpha|^2({V'}^n(k))_{RL}+|\beta|^2 ({V'}^n(k))_{LR} \notag \\
    &= \{ (|\alpha|^2-r)e^{ik}+(|\beta|^2-r)e^{-ik} \}\zeta_n(k)+\zeta_{n+1}(k). \label{eq:zeta1}
    \end{align}
Here we used the fact that the $n$-th power of   $\Theta=\begin{bmatrix} m_{11} & m_{12} \\ m_{21} & m_{22} \end{bmatrix}$ which has distinct eigenvalues can be described by 
    \[ \Theta^n = \begin{bmatrix} \zeta_{n+1} -m_{22}\zeta_n & m_{12}\zeta_n \\ m_{21}\zeta_n & \zeta_{n+1}-m_{11}\zeta_n \end{bmatrix} \]
in the last equality, 
where 
\[\zeta_n(k)=\frac{\lambda_+^n(k)-\lambda_-^n(k)}{\lambda_+(k)-\lambda_-(k)}\]
in the present case. 
We expand the eigenvalues by 
    \begin{align} 
    \lambda_\pm (k/\sqrt{n}) 
    &= r\left(1-\frac{k^2}{2n}\right)\pm \sqrt{r^2\left(1-\frac{k^2}{2n}\right)^2-(2r-1)} \notag \\
    &\sim r\left(1-\frac{k^2}{2n}\right) \pm \left( (1-r)-\frac{1}{2(1-r)} \frac{r^2k^2}{n} \right) \notag \\
    &= 1-\frac{r}{2(1-r)}\frac{k^2}{n},\;\;2r-1+\frac{r^2+r-1}{2 (1-r)} \frac{k^2}{n} \label{eq:exp1}
    \end{align}
for large $n$. 

Inserting (\ref{eq:exp1}) into (\ref{eq:zeta1}), we have 
    \begin{align*}
        \lim_{n\to^\infty}\hat{\mu}_n^{(0,0)}(k/\sqrt{n})
            & = \lim_{n\to\infty} 2(1-r) \zeta_n(k/\sqrt{n}) \\
            & = e^{-rk^2/(2(1-r))} \\
            & = \int_{-\infty}^{\infty} \frac{e^{-w^2/(2\sigma^2)}}{\sqrt{2\pi\sigma^2}} e^{ikw} dw, 
    \end{align*}
where $\sigma^2=r/(1-r)$. Note that the limit of the characteristic function with the diffusive scaling is independent of the initial state. 
Therefore the following central limit theorem holds: 
\begin{theorem}
\label{theorem-limit-OQRW}
    \[ \lim_{n\to \infty}\sum_{x\leq \sqrt{n}s}m_n(x)
    =\int_{-\infty}^{s} \frac{e^{-w^2/(2\sigma^2)}}{\sqrt{2\pi\sigma^2}}dw. \]
\end{theorem}
This theorem agrees with the limit of the correlated random walk~\cite{Konno_Cor} because the time recursion (\ref{eq:corRWrec}) is essentially the same as that of correlated random walk.
\section{Spectral analysis for $(s,t)\in \{(-1,0),(0,1)\}$ case}
\subsection{Eigenvalues of $\hat{W}_{s,t}$}
The operator in the Fourier space $\hat{W}_{s,t}(k)$ for $|s-t|=1$ is expressed by
    \[ \hat{W}_{s,t}(k)\cong \begin{bmatrix} V(k) & Q\otimes \bar{P} \\ P\otimes \bar{Q} & V(k) \end{bmatrix}. \]
Here if $(s,t)=(-1,0)$, then the computational basis are labeled by 
    \[ (0;LL),(0;LR),(0;RL),(0;RR),(-1;LL),(-1;LR),(-1;RL),(-1;RR); \]
while $(s,t)=(0,1)$, then the computational basis are labeled by 
    \[ (1;LL),(1;LR),(1;RL),(1;RR),(0;LL),(0;LR),(0;RL),(0;RR). \]
In what follows, we consider the case 
for 
\[ H=P+Q=\frac{1}{\sqrt{2}}\begin{bmatrix} 1 & 1 \\ 1 & -1\end{bmatrix}. \] 
Then we can compute that the eigenvalues $\hat{W}_2{s,t}(k)$ are $0$ with multiplicity $2$, and all the solutions of the following two cubic equations. 
    \begin{align*}
        2\lambda^3+(1-2c(k))\lambda^2-1 &= 0, \\
        2\lambda^3-(1+2c(k))\lambda^2+1 &= 0,
    \end{align*}
where $c(k)=\cos k$. 
The solutions are represented by the following multiple valued functions, respectively by Cardano's formula.
    \begin{align}
        \lambda_1(k) &= \frac{1}{6}\left\{ (2c(k)-1)+\frac{(2c(k)-1)^2}{\eta^{1/3}(c(k))} +\eta^{1/3}(c(k))\right\}, \label{eq:lambda1}\\
        \lambda_2(k) &= \frac{1}{6}\left\{ (2c(k)+1)+\frac{(2c(k)+1)^2}{\zeta^{1/3}(c(k))} +\zeta^{1/3}(c(k))\right\}, \label{eq:lambda2} 
    \end{align}
where
\begin{align*} 
\eta(r) &= 53+6r-12r^2+8r^3+6\sqrt{6}\sqrt{13+3r-6r^2+4r^3}, \\ 
\zeta(r) &= -53+6r+12r^2+8r^3+6\sqrt{6}\sqrt{13-3r-6r^2-4r^3}.
\end{align*}
The cubic roots $\eta^{1/3}(c(k))$ and $\zeta^{1/3}(c(k))$ become multiple functions taking three values, respectively.  We will show that {\it one}   eigenvalue from (\ref{eq:lambda1}) gives the Gaussian mode localized around the origin while {\it two} eigenvalues from (\ref{eq:lambda2}) give a ballistic mode like a wave packet. On the other hand, we will show that the effect of the rest of the eigenvalues disappear exponentially fast with respect to the time step. 

We put $c(k)=1-\epsilon$ for small $\epsilon \ll 1$ because we consider $\mu_n^{(s,t)}(k'/n^\sigma)$ for large $n$ and find an appropriate scaling order $\sigma$. Note that if $k=k'/n^\sigma$ then $\epsilon={k'}^2/(2n^{2\sigma})+O(1/n^{4\sigma})$. 
\subsubsection{$\lambda_1(k)$}
The expansion of the square root part of $\eta$ for $r=1-\epsilon$  is computed  by \begin{align*}
\sqrt{13+3r-6r^2+4r^3} = \sqrt{14}-\frac{3}{2\sqrt{14}} \epsilon+O(\epsilon^2).
\end{align*}
Inserting it into $\eta(1-\epsilon)$, we obtain
\[ \eta(1-\epsilon) = (55+12\sqrt{21})-(6+9\sqrt{3/7})\epsilon+O(\epsilon^2). \]
The cubic root $\eta(1-\epsilon)$ can be described by using the relation  $|(55+12\sqrt{21})^{1/3}|=5/2+\sqrt{21}/2$ as follows: 
\begin{align*}
    \eta(1-\epsilon)^{1/3} &= \left\{(5/2+\sqrt{21}/2)-\frac{2+3\sqrt{21}/7}{(5/2+\sqrt{21}/2)^{2}}\epsilon\right\}\omega^{j}, \\
    \eta(1-\epsilon)^{-1/3} &= \left\{(5/2-\sqrt{21}/2)+\frac{2+3\sqrt{21}/7}{(5/2+\sqrt{21}/2)^4}\epsilon\right\}\omega^{-j}
\end{align*}
for $j=0,1,2$, where $\omega=e^{2\pi i/3}$. 
Then the first terms of $\lambda_1(k)$  for $j=0,1,2$ are 
\[ \mathrm{first\;term\;of\;}\lambda_1(\delta)=\begin{cases} 1 & \text{: $j=0$,}\\ -\frac{1}{4}(1+i\sqrt{7}) & \text{: $j=1$,}\\ -\frac{1}{4}(1-i\sqrt{7}) & \text{: $j=2$,} \end{cases} \]
where $\delta^2/2=\epsilon$. 
Note that $|-(1\pm i\sqrt{7})/4|<1$. Then $\lambda_1(\delta)^n\to 0$ for $j=1,2$ which implies that the overlap to the eigenspaces for $j=1,2,$ decays exponentially with respect to the time step $n$.  
In the next, we focus on the second term of $\lambda_1$ in the case for $j=0$. 
\begin{align*}
    \lambda_1(\delta)
     &= 1+\frac{1}{6}\left\{ -2-\frac{4}{(55+12\sqrt{21})^{1/3}}-\frac{2+\frac{3\sqrt{21}}{7}}{(55+12\sqrt{21})^{2/3}}+\frac{2+\frac{3\sqrt{21}}{7}}{(55+12\sqrt{21})^{4/3}} \right\} \epsilon+O(\epsilon^2) \\
     &= 1-\frac{1}{2}\epsilon+O(\epsilon^2).
\end{align*}
Here we used that $(55+12\sqrt{21})^{\pm 1/3}\in \mathbb{R}$ are equal to $(5\pm \sqrt{21})/2$ and they are the solutions of $x^2-5x+1=0$ in the second equality. 
Therefore we have the following lemma.
\begin{lemma}
\label{lemma-lambda-1}
    \begin{equation}
        \lim_{n\to\infty}\lambda_1^n(k/\sqrt{n})= \begin{cases}e^{-k^2/4} & \text{: $j=0$,}\\
        0 & \text{: $j=1,2$.}\end{cases}
    \end{equation}
\end{lemma}
\subsubsection{$\lambda_2(k)$}
The expansion of the square root part of $\zeta$ for $r=1-\epsilon$ is quite different from $\eta$ as follows: 
\begin{align*}
\sqrt{13-3r-6r^2-4r^3} = 3\sqrt{3} \sqrt{\epsilon}+O(\epsilon),
\end{align*}
which will produce the ballistic behavior. Inserting it into $\zeta(1-\epsilon)$, we obtain
\[ \zeta(1-\epsilon) = -27+58\sqrt{2}\sqrt{\epsilon}+O(\epsilon). \]
The cubic root $\zeta(1-\epsilon)$ can be described by  
\begin{align*}
    \zeta(1-\epsilon)^{1/3} &= \left( -3+2\sqrt{2}\epsilon \right)\omega^{j}, \\
    \zeta(1-\epsilon)^{-1/3} &= \left( -\frac{1}{3}-\frac{2\sqrt{2}}{9}\sqrt{\epsilon} \right)\omega^{j}
\end{align*}
for $j=0,1,2$. 
Then the first terms of $\lambda_2(k)$  for $j=0,1,2$ are 
\[ \mathrm{first\;term\;of\;}\lambda_2(\delta)=\begin{cases} -1/2 & \text{: $j=0$,}\\ 1 & \text{: $j=1,2$,} \end{cases} \]
where $\delta^2/2=\epsilon$. 
Then the overlap to the eigenspaces for $j=0$ decays exponentially with respect to the time step $n$ in the same reason as the $\lambda_1$ case.  
Computing until the second term for the expansions of $\lambda_2(\delta)$ in the cases of $j=1,2$, 
we obtain  
    \begin{equation}\label{eq:secondorder}
        \lambda_2(\delta) = 
        \begin{cases} 
        0 & \text{: $j=0$,}\\
        1+i\sqrt{\frac{2}{3}}\sqrt{\epsilon}-\frac{4}{9}\epsilon & \text{: $j=1$,}\\
        1-i\sqrt{\frac{2}{3}}\sqrt{\epsilon}-\frac{4}{9}\epsilon & \text{: $j=2$.}
        \end{cases}
    \end{equation}
Therefore we have the following two stages of the limits: 
\begin{lemma}
\label{lemma-lambda-2}
For $j=0$, we have $\lambda_2^{n}(k/n^{\sigma})=0$ for any $\sigma>0$. 
On the other hand, for $j=1,2$, we obtain the following statements.
\begin{enumerate}
    \item Ballistic scaling
\begin{equation}
        \lim_{n\to\infty}\lambda_2^{n}(k/n)
        =\begin{cases}
        e^{ik/\sqrt{3}} & \text{: $j=1$,}\\
        e^{-ik/\sqrt{3}} & \text{: $j=2$.}
        \end{cases}
    \end{equation}
    \item Diffusive scaling
\begin{equation}
        \lim_{n\to\infty}  e^{\mp ik\sqrt{n}/\sqrt{3}}\lambda_2^{n}(k/\sqrt{n})
        =
        e^{-\frac{2k^2}{9}} \;\mathrm{for}\; j=1,2.
    \end{equation}    
\end{enumerate}
\end{lemma}    
\subsection{Eigenprojectoins}
In the previous section, we have obtained asymptotic behavior of the eigenvalues. In general, the sigularity of the eigenprojection is stronger than that of eigenvalues, in particular, around an unperturbed eigenvalue with some  multiplicity. Indeed the eigenvalue $1$ for the unperturbed operator $\hat{W}_{s,t}:=\hat{W}_{s,t}(0)$ has the multiplicity $3$; see Proposition~\ref{prop:eigenvec}. 
We need different analytical tool for this from the previous section, which is a perturbation theory of linear operators~\cite{Kato1982}. 

Let the perturbed operator $\hat{W}_{s,t}(\delta)$ with small $\delta$ be  expanded by 
\[\hat{W}_{s.t}(\delta)=\hat{W}_{s,t}+ T^{(1)}\delta + T^{(2)}\delta^2 +\cdots.\] 
The unperturbed operator $\hat{W}_{s,t}$ is no longer a normal operator but it has the following nice properties: 
\begin{enumerate}[(i)]
\item $\hat{W}_{s,t}$ is semi-simple in the sence of representation matrices. (Proposition~\ref{prop-semisimple}) 
\item The eigenprojection of $1\in \mathrm{Spec}(\hat{W}_{s,t})$; $\Pi$, is an orthogonal projection, that is, $\Pi=\Pi^*$. (Lemma~\ref{lem:orthogonal})  
\item $\Pi T^{(1)}\Pi$ is a normal operator and does not have any multiple eigenvalues in the range of $\Pi$; $R(\Pi)$.  (Proposition~\ref{prop:PTP})
\end{enumerate}
The dimension of $\mathrm{R}(\Pi)$ is three by Proposition~\ref{prop:eigenvec}. Let simple eigenvalues on $R(\Pi)$ of $\Pi T^{(1)}\Pi$ be $\lambda_j^{(1)}$ $(j=1,2,3)$.  
By the reduction process~\cite{Kato1982} of $\hat{W}_{s,t}(\delta)$ to $\tilde{T}^{(1)}(\delta)$ (see Appendix for its definition) with $\tilde{T}^{(1)}(0)=\Pi T^{(1)}\Pi$, the cardinality of $(1+\lambda_j^{(1)} \delta)$-group of  $\hat{W}_{s,t}(\delta)$ is just $1$ 
since the splitting from $\lambda_j^{(1)}$ does not happen by the small  perturbation $\delta$ because of property (iii) (see Lemma~\ref{lemma-Kato} and Remark~\ref{remark-Kato} for more detail). 
Note that such an eigenprojection of a simple eigenvalue is continuous~\cite{Kato1982}.  
Then by the property (ii), each perturbed eigenprojection $\Pi_j(\delta)$ $(j=1,2,3)$ converges to the orthogonal projection onto its eigenvector of $\Pi T^{(1)}\Pi$, say $v_jv_j^*$. The eigensystem of $\Pi T^{(1)}\Pi$ can be computed in Appendix explicitly for the Hadamard walk. 
\begin{theorem}\label{thm:eigenprojection}
Let $U^{(2)}$ be the quantum walk defined in (\ref{eq:U2}) with the Hadamard quantum coin.
Eigenvalues of  $\hat{W}_{s,t}(\delta)=\chi_{s,t}\mathcal{F}U^{(2)}\mathcal{F}^*\chi^*_{s,t}$ closed to the unit circle in the complex plain are
\begin{equation*}
1-\delta^2/4+o(\delta^2),\quad 1 \pm \frac{i}{\sqrt{3}}\delta - (2/9)\delta^2+o(\delta^2)
\end{equation*}
as $\delta\to 0$.
The corresponding perturbed eigenprojections of 
the three eigenvalues converge to $v_1v_1^*$, $v_2v_2*$ and $v_3v_3^*$, respectively, as $\delta\to 0$, where
\begin{align*}
    v_1 &= (-1/2, 0,0,-1/2,1/2,0,0,1/2)^\top, \\
    v_2 &= \frac{1}{\sqrt{12(2-\sqrt{3})}}(2-\sqrt{3},0,1-\sqrt{3},1,2-\sqrt{3},1-\sqrt{3},0,1)^\top, \\
    v_3 &= \frac{1}{\sqrt{12(2+\sqrt{3})}}(2+\sqrt{3},0,1+\sqrt{3},1,2+\sqrt{3},1+\sqrt{3},0,1)^\top.
\end{align*}
\end{theorem}
\begin{proof}
Note that Proposition~\ref{prop:eigensystem} gives the eigenvalues  of $\hat{W}_{s,t}(\delta)$ around the unit circle in the complex plain up to the first order with respect to the small perturbation $\delta$. We applied (\ref{eq:secondorder}) for the second order. 
For the eigenprojection, this is a direct consequence of Proposition~\ref{prop:eigensystem}.  
\end{proof}
\section{Limit theorems for $(s,t)\in\{(1,0),(0,1)\}$ case}
Our purpose of this section is to obtain limit theorems with respect to the measure $\mu_n^{(s,t)}$ in Definition~\ref{def:measure}. 
Let $\lambda_j(\delta)$,  $P_j(\delta)$ and $D_j(\delta)$  be a perturbed eigenvalue, eigenprojection and eigennilponent of $\hat{W}_{s,t}(\delta)$, respectively. Then $\hat{W}_{s,t}(\delta)$ is described by 
    \[ \hat{W}_{s,t}(\delta)=\sum_{j}\lambda_j(\delta)P_j(\delta)+D_j(\delta). \]
Let us prepare the fact on the eigennilponent of $\hat{W}_{s,t}$. We will use this for the proof of the limit theorem of $\mu_n^{(s,t)}$. 
\begin{lemma}\label{lem:D}
Every eigennilponents $D_j(\delta)=0$ for sufficient small $|\delta|$. \end{lemma}
\begin{proof}
By (\ref{eq:eigeneq}) in Appendix, the eigenvalues of unperturbed operator $\hat{W}_{s,t}$ except $1$ and $0$ are simple. Then $D_j(\delta)=0$ for their corresponding perturbed eigenvalues when $|\delta|$ is sufficiently small. On the other hand, for the unperturbed eigenvalue $1$ which has multiplicity $3$, by the reduction process~\cite{Kato1982}, the eigenproblem of $\hat{W}_{s,t}(\delta)$ is reduced to that of the perturbed operator of $\Pi T^{(1)}\Pi$. Since we have observed that $\Pi T^{(1)}\Pi$ is a normal operator, it holds that the corresponding eigenprojections and eigennilponents, say $j=1,2,3$, are    
$\Pi_j(\delta)=\Pi_j^{(1)}(\delta)$ with $\Pi_j^{(1)}(0)=v_jv_j^*$ and  $D_j(\delta)=0$ by Theorem~\ref{thm:eigenprojection}, respectively for small $|\delta|$. Finally around the unperturbed eigenvalue $0$ for $\hat{W}_{s,t}(\delta)$, it is easy to see that the perturbed eigenvectors are generated by $\delta_s|LR\rangle$ and $\delta_t|RL \rangle$ which are independent of $\delta$. Thus the splitting never occur. Then we have $D_j(\delta)=0$ for any $j$. 
\end{proof}

We obtain the following limit theorem for $\mu_n^{(s,t)}$. 
\begin{theorem}
\label{theorem-limit-behavior}
The Dirichret-cut quantum walk for $s=-1,t=0$ with the initial state $\delta_{0,0}(x) \varphi_0 \otimes \overline{\varphi_0}$, where $\varphi_0=[g_1\;g_2]^\top\in \mathbb{C}^2$ is a unit vector, has 
the diffusive mode described by $N(0,1/2)$ and also 
the two ballistic modes whose spreading speeds are $\pm 1/\sqrt{3}$ and width are described by the diffusion $N(0,4/9)$, respectively: more precisely, for any $a,b\in \mathbb{R}$ $(a<b)$, 
    \begin{align}
        \lim_{n\to \infty}\sum_{an\leq x\leq bn} \mu_n^{(s,t)}(x) 
        &= \int_{a}^{b} \left(c_-\;\delta_{-1/\sqrt{3}} (y)+c_0\;\delta_0(y)+c_+\;\delta_{1/\sqrt{3}}(y)\right) dy \label{eq:largenumber}
    \end{align}
and 
    \begin{align}
        \lim_{n\to \infty}\sum_{a\sqrt{n}\leq x\leq b\sqrt{n}} \mu_n^{(s,t)}(x) 
        &= c_0\int_{a}^{b} \frac{e^{-y^2}}{\sqrt{\pi}} dy, \label{eq:clt+}\\
        \lim_{n\to \infty}\sum_{a\sqrt{n}\leq x\mp n/\sqrt{3} \leq b\sqrt{n}} \mu_n(x) 
        &= c_{\pm}\int_{a}^{b} \frac{e^{-y^2/(8/9)}}{\sqrt{(8/9)\pi}} dy,\label{eq:clt-}
    \end{align}
where 
    \begin{align}
        c_0 = 1/2;\;
        c_\pm = \frac{(2\mp\sqrt{3})|g_1|^2+(1\mp\sqrt{3})\bar{g}_1g_2+|g_2|^2}{2\; (3\mp\sqrt{3})}.
    \end{align}
\end{theorem}
\begin{proof}
The initial state in the real space $\ell^2( \mathbb{Z}\times \{s,\dots,t\} ; \mathbb{C}^4)$ is denoted by    $\phi_0^{(s,t)}(x,y)=\delta_{(0,0)}(x,y)\varphi_0'\otimes \overline{\varphi_0'}$ for some unit vector $\varphi_0'=[g_1\;g_2]^\top\in \mathbb{C}^2$.   The Fourier transform of $\phi_0^{(s,t)}$ is described by 
$\hat{\phi}_0^{(s,t)}(k;v)=\delta_0(v) \varphi_0'\otimes \overline{\varphi_0'}$ with  
    \[ \hat{\phi}_0^{(s,t)} = u + \alpha_1v_1 + \alpha_2v_2+\alpha_3v_3\]
for some $\alpha_j \in \mathbb{C}$ $(j\in\{1,2,3\})$ and $u\in R(\Pi)^\perp$. 
Here $\alpha_j=\langle v_j,\hat{\phi}_0^{(s,t)} \rangle$, 
$u=\sum_{j=1}^3(1-v_jv_j^*)\hat{\phi}_0^{(s,t)}$.  

If the unperturbed eigenvalue of $\hat{W}_{s,t}$; $\lambda_j$, satisfies $|\lambda_j|<1$, then since the perturbed eigenvalue $\lambda_j(\delta)$ is continuous with respect to $\delta$, we have $|\lambda_j(\delta))|<1$ for small $\delta$. 
Note that the corresponding eigenprojection $\Pi_j(\delta)$ can be evaluated by
    \[ ||\Pi_j(\delta)||< c \delta^{-\tau} \]
for some constant $c>$ and $\tau>0$~\cite{Kato1982}. 
Therefore putting $\delta=k/n^{\theta}$, we obtain that if $|\lambda_j|<1$, then $\lambda_j(\delta)^n\Pi_j(\delta)\to 0$ as $n\to\infty$. 
Therefore combining it with Lemma~\ref{lem:D}, we have 
\begin{align}
    \lim_{n\to\infty}\hat{W}_{s,t}^n(\delta)
    &= \lim_{n\to\infty}\sum_{j}\lambda_j^n(\delta) \Pi_j(\delta) \notag\\
    &= \lim_{n\to\infty}\sum_{j=1}^3 \lambda_j^n(\delta) v_jv_j^* \notag\\
    &= \lim_{n\to\infty} e^{-\delta^2 n/4}v_1v_1^* + e^{i\delta n/\sqrt{3}-\frac{1}{2}(2\delta/3)^2n}v_2v_2^*+e^{-i\delta n/\sqrt{3}-\frac{1}{2}(2\delta/3)^2n}v_3v_3^*. \label{eq:key}
\end{align}
Here we used the expansions of $\lambda_1(\delta)$, $\lambda_2(\delta)$ and $\lambda_3(\delta)$ in Theorem~\ref{thm:eigenprojection} in the last equality.
Then if we choose $\theta=1$, that is, $\delta=k/n$, we obtain
\begin{equation}
    \lim_{n\to\infty}\hat{W}_{s,t}^n(k/n)
    =v_1v_1^* + e^{ik/\sqrt{3}}v_2v_2^*+e^{-ik/\sqrt{3}}v_3v_3^*,
\end{equation}
which implies 
    \begin{align} 
    \lim_{n\to\infty}\hat{\mu}_n^{(s,t)}(k/n)
    &= \alpha_1\langle q_0,v_1\rangle 
    + \alpha_2\langle q_0,v_2\rangle e^{ik/\sqrt{3}}
    +\alpha_3\langle q_0,v_3\rangle e^{-ik/\sqrt{3}} \notag \\
    &= \alpha_1  
    + \frac{\alpha_2}{\sqrt{2}} e^{ik/\sqrt{3}}
    + \frac{\alpha_3}{\sqrt{2}} e^{-ik/\sqrt{3}} \label{eq:limit1}
    \end{align}
where $q_0(x)=\delta_0(x)[1\;0\;0\;1]^\top$ and 
    \begin{align*}
        \alpha_1 &= \langle v_1,\hat{\phi}_0^{(s,t)} \rangle = \frac{1}{2}, \\
        \alpha_2 &= \langle v_2,\hat{\phi}_0^{(s,t)} \rangle = \frac{(2-\sqrt{3})|g_1|^2+(1-\sqrt{3})\bar{g}_1g_2+|g_2|^2}{\sqrt{2}\; (3-\sqrt{3})}, \\
        \alpha_3 &= \langle v_2,\hat{\phi}_0^{(s,t)} \rangle = \frac{(2+\sqrt{3})|g_1|^2+(1+\sqrt{3})\bar{g}_1g_2+|g_2|^2}{\sqrt{2}\;(3+\sqrt{3})}. 
    \end{align*}
Then taking the Fourier inversion to RHS of (\ref{eq:limit1}), we have 
    \begin{align*}
    \lim_{n\to\infty}
    \sum_{x\in \mathbb{Z}}\mu_n^{(s,t)}(x)e^{ikx/n}
        &= \int_{-\infty}^{\infty} \left(-\alpha_1\delta_0(y)+\frac{\alpha_2}{\sqrt{2}}\delta_{1/\sqrt{3}}(y)+\frac{\alpha_3}{\sqrt{2}}\delta_{-1/\sqrt{3}}(y)\right) e^{iky} dy.
    \end{align*}
Then the ballistic scaling gives the three mass points at the left and right sides and the center of the stripe, which seems to correspond to so called law of large number in the classical theory. The delta measure at the origin corresponds to localization in the literature of quantum walk's study~\cite{Konno2008b}. It completes the proof of (\ref{eq:largenumber}).

On the other hand, if we choose $\theta=1/2$, we obtain another kind of limit theorem. This gives an observation how the walker diffuses around each position, $0,\pm n/\sqrt{3}$, which looks like so called the central limit theorem for asymmetric random walk in the classical theory. 
First around the center of the stripe, we have  
\begin{align*}
    \lim_{n\to\infty} \sum_{a\sqrt{n}\leq x'\leq b\sqrt{n}} \mu_n^{(s,t)}(x')
    &= \lim_{n\to\infty} \sum_{a\sqrt{n}\leq x'\leq b\sqrt{n}} \int_{-\pi}^\pi \hat{\mu}_n(k)e^{-ikx'} \frac{dk}{2\pi} \\
    &= \lim_{n\to\infty} \sum_{a\sqrt{n}\leq x'\leq b\sqrt{n}}\frac{1}{\sqrt{n}} \int_{-\sqrt{n}\pi}^{\sqrt{n}\pi} \hat{\mu}_n(k/\sqrt{n})e^{-ikx'/\sqrt{n}} \frac{dk}{2\pi} \\
    &= \lim_{n\to\infty} \sum_{a\sqrt{n}\leq x' \leq b\sqrt{n}}\frac{1}{\sqrt{n}} \int_{-\sqrt{n}\pi}^{\sqrt{n}\pi} \alpha_1\langle q_0,v_1 \rangle e^{-k^2/4}e^{-ikx'/\sqrt{n}} \frac{dk}{2\pi} \\
    &= \lim_{n\to\infty} \sum_{a\sqrt{n}\leq x' \leq b\sqrt{n}}\frac{1}{\sqrt{n}} \frac{1}{2}\frac{e^{-(x'/\sqrt{n})^2}}{\sqrt{\pi}} \\
    &= c_0 \int_{a}^{b} \frac{e^{-y^2}}{\sqrt{\pi}} dy, \;\;\;\; \mathrm{for\;any\;}a<b, 
\end{align*}
where $c_0=1/2$. 
Here in the third equality, we inserted (\ref{eq:key}) into $\hat{\mu}_n(k/\sqrt{n})$ by 
    \[ \hat{\mu}_n(k/\sqrt{n})=e^{-k^2/4}\alpha_1\langle q_0,v_1\rangle +e^{i\sqrt{n/3}-\frac{1}{2}(2k/3)^2}\alpha_2\langle q_0,v_2\rangle + e^{-i\sqrt{n/3}-\frac{1}{2}(2k/3)^2} \alpha_3\langle q_0,v_3\rangle\]
and applied the Riemann-Lebesgue Lemma to the second and third terms. 
Secondly, around the right side of the stripe; $n/\sqrt{3}$, we have
\begin{align*}
    \lim_{n\to\infty} 
    \sum_{a\sqrt{n}\leq x'-n/\sqrt{3} \leq b\sqrt{n}} \mu_n^{(s,t)}(x')
    &= \lim_{n\to\infty} \sum_{a\sqrt{n}\leq x' \leq b\sqrt{n}}\frac{1}{\sqrt{n}} \int_{-\sqrt{n}\pi}^{\sqrt{n}\pi} \left\{e^{-ik\sqrt{n/3}}\hat{\mu}_n(k/\sqrt{n})\right\} e^{-ikx'/\sqrt{n}}\frac{dk}{2\pi} \\
    &= \lim_{n\to\infty} \sum_{a\sqrt{n}\leq x' \leq b\sqrt{n}}\frac{1}{\sqrt{n}} \left(\alpha_1\langle q_0,v_1 \rangle \frac{e^{-(x'/\sqrt{n})^2/(8/9)}}{\sqrt{ (8/9)\pi}}\right) \\
    &= c_+\int_{a}^{b} 
    \frac{e^{-y^2/(8/9)}}{\sqrt{(8/9)\pi}} dy.  
\end{align*}
Here we applied the Riemann-Lebesgue Lemma to the second and third terms of (\ref{eq:key}) and 
    \[\lim_{n\to\infty}e^{-ikn/\sqrt{3}}\hat{W}_{s,t}^n(k/\sqrt{n})v_2 = e^{-2k^2/9}v_2 \] 
in the second equality. Then it completes the proof of (\ref{eq:clt+}). 
For the proof of (\ref{eq:clt-}), in a similar fashion to the proof of (\ref{eq:clt+}), around the left side of the stripe, we obtain 
    \begin{equation*}
        \lim_{n\to\infty} 
    \sum_{a\sqrt{n}\leq x'+n/\sqrt{3} \leq b\sqrt{n}} \mu_n^{(s,t)}(x')
    = c_-\int_{a}^{b} 
    \frac{e^{-y^2/(8/9)}}{\sqrt{(8/9)\pi}} dy.
    \end{equation*}
Then we have obtained the desired conclusion. 
\end{proof}

\section{Numerical Observations}

\input{numerics}


\section{Summary and Discussion}
We proposed an interpolation walk model associated with a parameters $t<0<s$ between an open quantum random walk and a quantum walk on the one-dimensional lattice. The parameters adjust the width of the stripe in the two-dimensional lattice which is parallel with $x$-axis. To see the crossover, we defined a measure, which may take a complex value but reproduce the distributions of the open quantum random walk $((t,s)=(0,0))$ and the quantum walk ($(t,s)=(-\infty,\infty)$).   
In this paper, we analytically concentrated on the case for $|s-t|=1$ which expected to slightly run off the behavior of the open quantum random walk. 
We obtained a limit theorem which shows a coexistence of localization and ballistic spreading using the Kato perturbation theories. 
Then we can see a typical behavior of quantum walks in such a small casting of the interference effect into the original open quantum random walk.  
However we futher analyse this limit theorem in more detail until the second order, we obtain the three wave packets located in the left edge, center and right edge, described by Gaussian distributions. Thus some properties of the open quantum random walk may still remain. These results are also supported by the numerical simulations. 

In numerical simulations, we also observe that as the width of the stripe increases, the number of islands increases in the `distribution'.  The height and the width of the two islands located in the left- and rightmost edges are suggested to converge to $O(n^{-2/3})$ and $O(n^{-1/3})$, for large $n$, respectively. On the other hand, by a detailed analysis around the singular points of (\ref{Konno}), which is the limit distribution for $(s,t)=(-\infty,\infty)$, the distribution around the singular points can be expressed by the Airy functions, and the region where the decay rate is estimated by non-linear order $O(n^{-2/3})$, is also  $O(n^{-1/3})$~\cite{Sunada_Tate2012}. Our proposed measure can be expected to be an  indicator to see how a quantum interference, decoherence affect the system in more detail.  
However we have not explored deeply this measure yet, for example, to show the observation by the numerical simulation that the presence of negative values of the measure makes the many islands  is still interesting future's problem. 

Finally, let us consider our model in the literature of the realization procedure of the open quantum random walk in Proposition~8.1~\cite{AttalEtAl2}. 
Let the total state space be the composition $\mathcal{H}\otimes \mathcal{K}_1 \otimes \mathcal{K}_2$. 
Here $\mathcal{H}$ is isomorphic to $\mathbb{C}^d$ with some constant $d\in\mathbb{N}$, and the computational basis of $\mathcal{K}_1$ are denoted by $\{|i\rangle \}_{i\in \mathbb{Z}}$, and $\mathcal{K}_2$ is the copy of $\mathcal{K}_1$. 
The matrix valued weight on $\mathcal{H}$ associated with the moving of an open quantum random walker from $j$ to $i$ is denoted by $B_k^j$ satisfying with $\sum_k {B_k^j}^*{B_k^j}=I$.
Let $E$ be a unitary operator on the total space described by
    \[ E=\sum_{i,j,k}{}^iE_j(k) \otimes |i\rangle\langle j|\otimes |k\rangle\langle k| \]
so that ${}^jE_1(k)=B_k^j$. 
Assume that the underlying network is one-dimensional lattice; that is, if $|j-k|>1$, then $B_k^j=0$. 
The equivalent realization procedure to the open quantum random walk is expressed as follows.  
First we set the initial state by $\rho^{(0)}:=\sum_k \rho_k \otimes |1\rangle\langle 1|\otimes |k\rangle\langle k|$. To this initial state, we act the following procedure.  
\begin{enumerate}
    \item an action of the unitary operator $U$.
    \item a decoherence on the computational basis of $\mathcal{K}_1$.
    \item an action of the swap operation on $\mathcal{K}_1\otimes \mathcal{K}_2$ so that $|\phi_1\rangle \otimes |\phi_2\rangle \mapsto |\phi_2\rangle \otimes |\phi_1\rangle$. 
    \item a refreshing of the system $\mathcal{K}_1$ to the state $|1\rangle\langle 1|$.  
\end{enumerate}
If we repeat this procedure until $n$ iterations replacing $\rho^{(0)}$ into the state changed by the above procedure at each iteration, then the final state is of the form  $\sum_{x}\rho_x^{(n)} \otimes |1\rangle \langle 1| \otimes  |x\rangle\langle x|$. The probability distribution on the one dimensional lattice of the open quantum random walk at time $n$ is $\mathrm{tr}(\rho_x^{(n)})$. 

In the decoherence step (2), if the state is described by  $\sum_{x\in\mathbb{Z}} B_k^x|\varphi\rangle \otimes |x \rangle \otimes |k\rangle$ after the unitary action step (1), then 
the state after the step (2) is changed as follows:  
    \[ \sum_{x\in\mathbb{Z}} B_k^x|\varphi\rangle \otimes |x \rangle \otimes |k\rangle
     \mapsto \sum_{x} B_k^x|\varphi\rangle\langle \varphi|{B_k^x}^* \otimes |x\rangle\langle x|\otimes |k\rangle\langle k|. \]
On the other hand, our model corresponds to the walk extending the decoherence step (2) 
by newly introducing parameters $s,t\in\mathbb{R}$ with $s\leq 0\leq t$ such that     
\[ \sum_{x\in\mathbb{Z}} B_k^x|\varphi\rangle \otimes |x \rangle \otimes |k\rangle
     \mapsto \sum_{s\leq x-y\leq t} B_k^x|\varphi\rangle\langle \varphi|{B_k^y}^* \otimes |x\rangle\langle y|\otimes |k\rangle\langle k|.  \]
After the $n$-th iteration of the procedure, the state is of the form $\sum_{s\leq x-y \leq t} \rho_{x,y}^{(n)}\otimes |1\rangle \langle 1| \otimes |x\rangle \langle y|$. 
The measure is described by $\mu_n^{(s,t)}(x) =\mathrm{tr}(\rho_{x,x}^{(n)})$.

Remark that if $s=t=0$, then the original decoherence step is recovered, while if $s=-\infty$ and $t=\infty$ and $\{B_j^i\}$  satisfy the restrictive condition (10) in \cite{AttalEtAl2}, then because the situation is equivalent to skipping the decoherence step, a unitary quantum walk is recovered (see Proposition~10.1~\cite{AttalEtAl2}). 

\noindent\\
\noindent {\bf Acknowledgments}. 
The authors thank to Ji\u{r}\'{i} Mary\u{s}ka,   Stanisla Skoup\'{y} for fruitful discussion. 
K.M. was supported by World Premier International Research Center Initiative (WPI), Ministry of Education, Culture, Sports, Science and Technology (MEXT), Japan and the grant-in-aid for young scientists No.~17K14235, Japan Society for the Promotion of Science. 
E.S. acknowledges financial supports from the Grant-in-Aid of
Scientific Research (C)  No.~19K03616, Japan Society for the Promotion of Science and Research Origin for Dressed Photon.

\bibliographystyle{jplain}

\end{small} 


\appendix
\section{Eigenspace of unperturbed operators and their perturbation}
\input{eigenspace_unperturbed}

\end{document}

%% file: numerics.tex
In this section, we summarize numerical observations for various $M$ including $M=1,2$ visualizing the connection of the nature among open quantum random walks and quantum walks through our proposed model.
Here we study the following objects to characterize the nature of dynamics:
\begin{itemize}
\item Distributions of walkers for each $M$, restricted to the diagonal $x=y$ (Section \ref{section-num-distribution}).
\item Effective time of interference of the boundary (Section \ref{section-num-boundary}).
\item Characterization of localized peaks (Section \ref{section-num-localized}).
\item Comparison between numerical results and analytic results (Section \ref{section-num-comparison}).
\end{itemize}

In the whole computations, we fix the initial state $\phi$ as
\begin{equation*}
\phi = \phi_{1}\otimes \phi_1,\quad \text{ where } \quad \phi_1 = \begin{pmatrix}
1/2\\
1/2
\end{pmatrix}.
\end{equation*}
Moreover, we restrict our attention of domains to
\begin{equation*}
(s,t) = \begin{cases}
(-(M-1)/2, (M-1)/2) & \text{ if $M$ is odd},\\
(-M/2, M/2-1) & \text{ if $M$ is even}
\end{cases}
\end{equation*}
and the corresponding density is denoted as
\begin{equation*}
\mu_{n,M}(x,y) := \begin{cases}
\mu_n^{(-(M-1)/2, (M-1)/2)}(x,y) & \text{ if $M$ is odd},\\
\mu_n^{(-M/2, M/2-1)}(x,y) & \text{ if $M$ is even},
\end{cases}
\end{equation*}
where
\begin{equation*}
\mu_n^{(s,t)}(x,y) = \begin{cases}
\langle LL | \phi_n^{(s,t)}(x,y)\rangle + \langle RR | \phi_n^{(s,t)}(x,y)\rangle & (x,y)\in D_{s,t},\\
0 & \text{otherwise}.\\
\end{cases}
\end{equation*}
We shall also use the following notation:
\begin{equation*}
\mu_{n,M}(x) := \mu_{n,M}(x,x).
\end{equation*}
Note that all figures in this section is normalized so that the all walkers are distributed inside $\bar x\in [-1,1]$.
More precisely, we draw all distributions of walkers as relationships between $\bar x \equiv x/n$ and $n \mu_{n,M}(\bar x) $ for a given time $n$.

\subsection{Distributions of walkers for each $M$, restricted to the diagonal $x=y$}
\label{section-num-distribution}
First we show asymptotic distribution of walkers for each $M$.
We first study typical behavior of the present model.
Figure \ref{fig-dynamics_sample} shows the distribution with relatively large $M$. 
In the present situation, namely the walker does not arrive at the boundary, the distribution is like the Hadamard walk.
A criterion for checking the agreement of the distribution with the Hadamard'd walk is to compare with the normalized limit function (cf. \cite{Konno2008b})
\begin{equation}
\label{Konno}
K(x) = \frac{1}{\pi (1-x^2)\sqrt{1-2x^2}}{\bf 1}_{[-1/\sqrt{2}, 1/\sqrt{2}]},
\end{equation}
where ${\bf 1}_{I}$ is the characteristic function whose support is the interval $I$.
Figure \ref{fig-dynamics_sample} (b)-(c) show the agreement of the (normalized) limit distrubution of our present model for large $M$ with $K(x)$ in the weak sense.

\begin{remark}
In Figures \ref{fig-dynamics_sample}-\ref{fig-dynamics_long_varM_2}, the scalings $(x,y) \mapsto (x/n, y/n)$ and $\phi_n \mapsto n\phi_n$ are applied for visibility.
The former is operated so that the support of walker density distribution is included in $[-1,1]$. 
The latter is then operated so that the integral of walker density over $[-1,1]$ is kept under the present transformation.
\end{remark}

\begin{figure}[htbp]
\begin{minipage}{0.33\hsize}
\centering
\includegraphics[width=6.0cm]{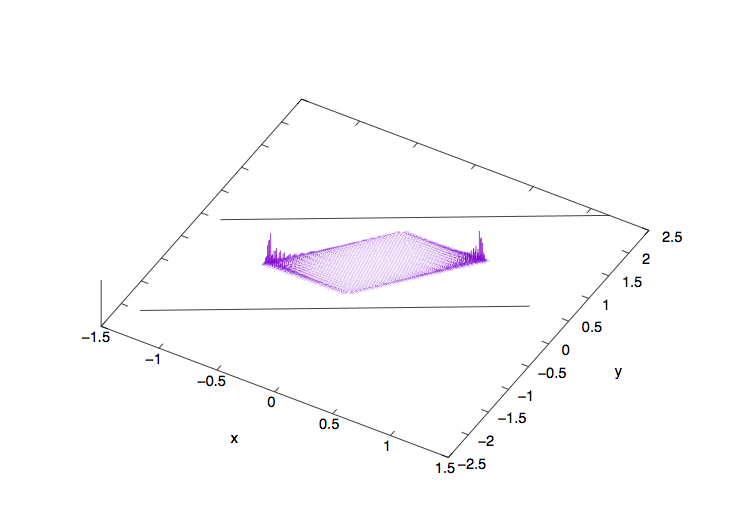}
(a)
\end{minipage}
\begin{minipage}{0.33\hsize}
\centering
\includegraphics[width=6.0cm]{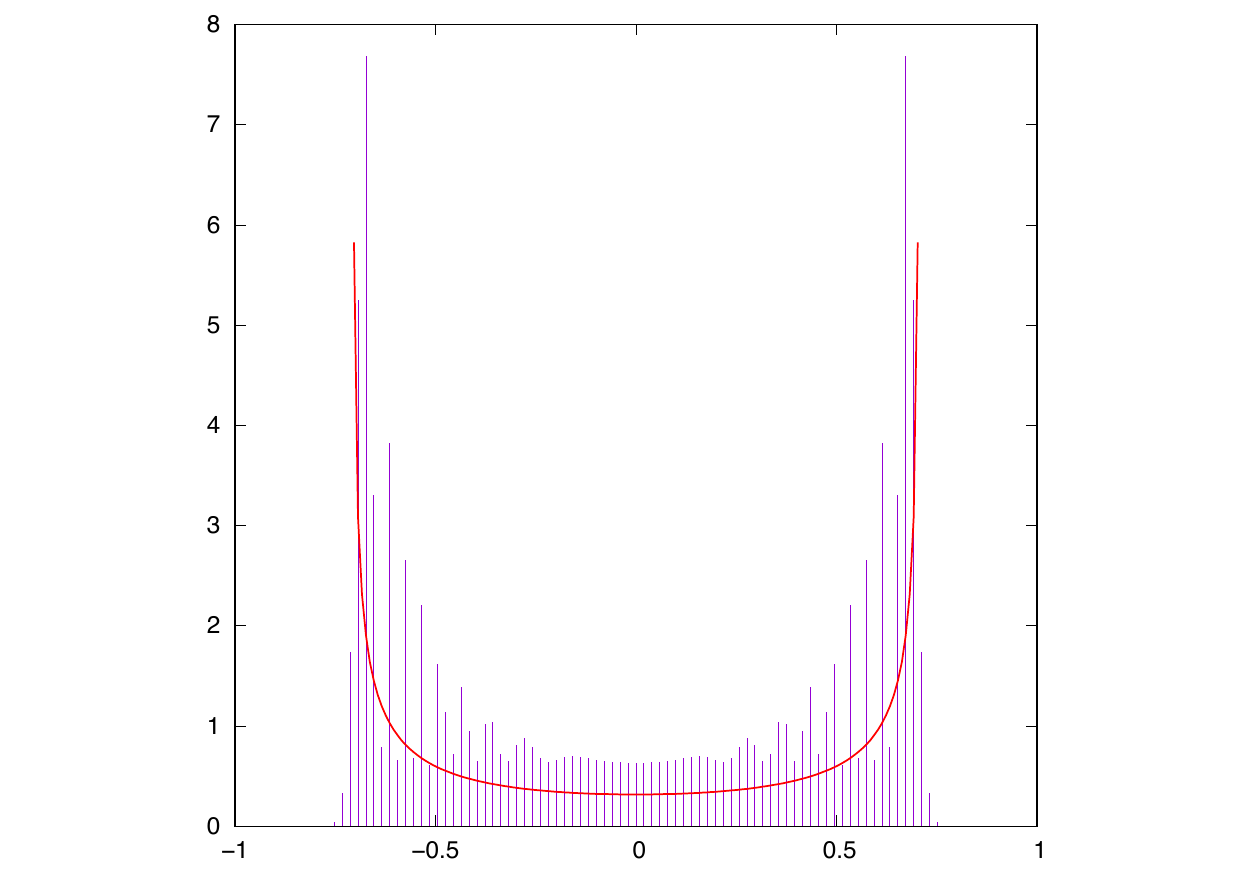}
(b)
\end{minipage}
\begin{minipage}{0.33\hsize}
\centering
\includegraphics[width=6.0cm]{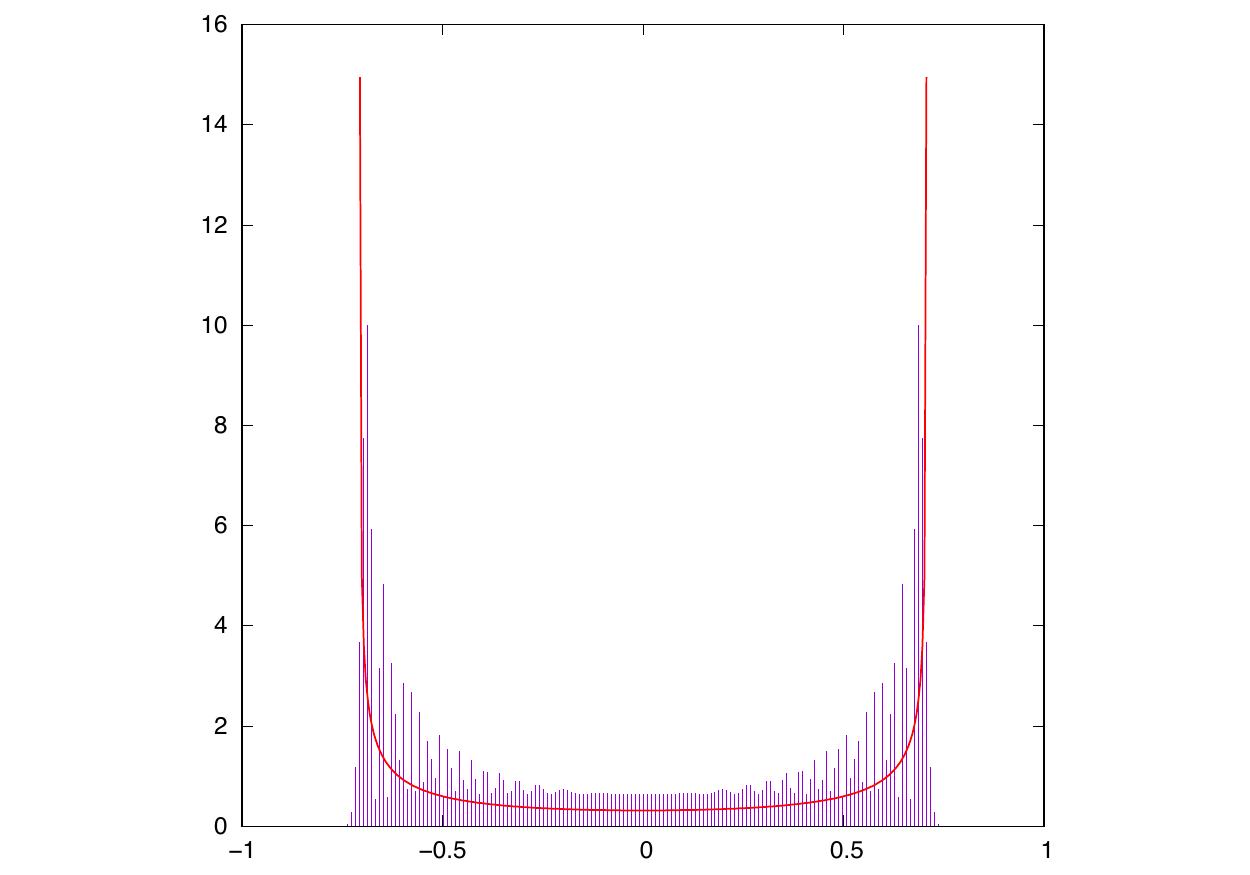}
(c)
\end{minipage}
\caption{Typical dynamics of the present model with large $M$}
\label{fig-dynamics_sample}
(a): $3$-dimensional plot of distribution under $n=100$ iterations and $M=201$. 
Black solid lines denote the boundary of the domain.
(b): Plot of $\mu_{100,201}(x)$ with $M=201$. 
(c): Plot of $\mu_{200,500}(x)$ with $M=500$. 
From (b) and (c), we see that the distribution of $\mu_{n,M}$ behaves like the well-known Hadamard walk.
The red curves denote the weak-limit distribution function $K(x)$ in (\ref{Konno}).
\end{figure}

Next we change the width $M$ and the iteration step $n$ so that $n$ is sufficiently larger than $M$, and compute the dynamics.
Figures \ref{fig-dynamics_sample_varM_1}-\ref{fig-dynamics_sample_varM_2} show the distribution of walkers $\mu_{n,M}(x)$ at each $x$ after $n=100$ iterations, restricted to the diagonal line $x=y$.
As indicated in Theorem \ref{theorem-limit-OQRW}, walkers distribute like random walks when $M=1$, which can be seen in Figure \ref{fig-dynamics_sample_varM_1}-(a).
When we change $M$ into $2$, the quantum-walk-like interaction comes into play, as seen in Figure \ref{fig-dynamics_sample_varM_1}-(b).
More precisely, the localized peaks arise in both sides of the central peak.
As $M$ increases, non-trivial distributions arise between the center and the size peak.
When $M=3$ (Figure \ref{fig-dynamics_sample_varM_1}-(c)), the height of peaks among the center and sides becomes significantly different, unlike the case $M=2$.
Moreover, we also observe the negative-value distribution when $M\geq 2$.
This tendency persists for all $M$ as far as we have computed, while it does not appear for small $n$.
Details are discussed in Section \ref{section-num-boundary}.
Finally, the distribution of walkers with positive density converges to that for the Hadamard walk as $M$ becomes larger and larger.
At the same time, we have also studied long-time behavior of these distributions.
Figures \ref{fig-dynamics_long_varM_1}-\ref{fig-dynamics_long_varM_2} show the distribution of walkers after $n=10000$ iterations with $M=1,2,3,5,10, 51$.
We see that the tendency of walker distributions is similar to the corresponding shorter time evolution results (Figures \ref{fig-dynamics_sample_varM_1}-\ref{fig-dynamics_sample_varM_2}), while we can see the behavior clearer.
For example, the side peaks arise when $M\geq 2$, but the heights of peaks look almost identical in the case $M=2$, while they are significantly different when $M\geq 3$.
As $M$ increases, the number of localized islands increases and accumulated to construct the Hadamard walk-like limiting behavior.

\begin{figure}[htbp]
\begin{minipage}{0.33\hsize}
\centering
\includegraphics[width=5.5cm]{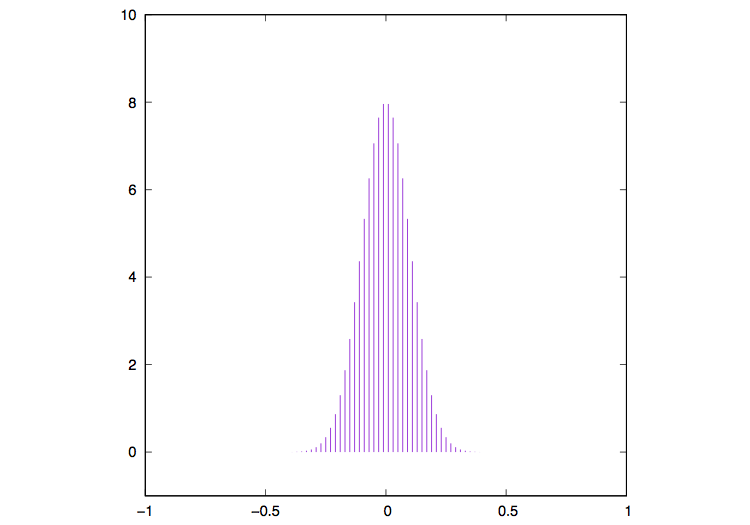}
(a)
\end{minipage}
\begin{minipage}{0.33\hsize}
\centering
\includegraphics[width=5.5cm]{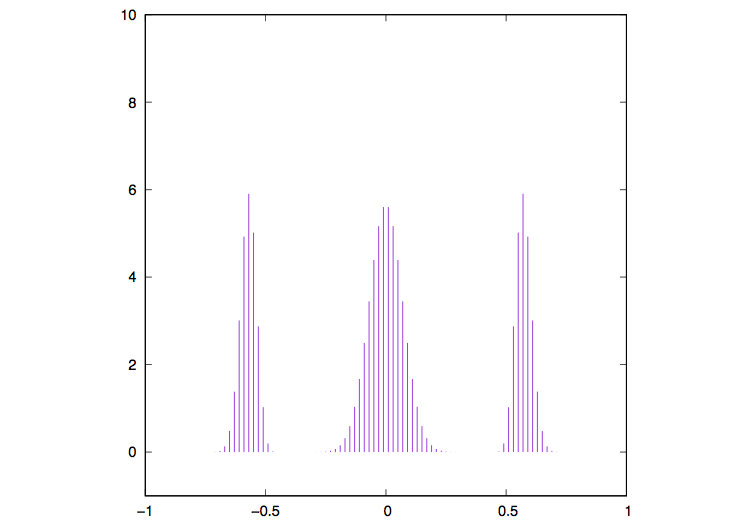}
(b)
\end{minipage}
\begin{minipage}{0.33\hsize}
\centering
\includegraphics[width=5.5cm]{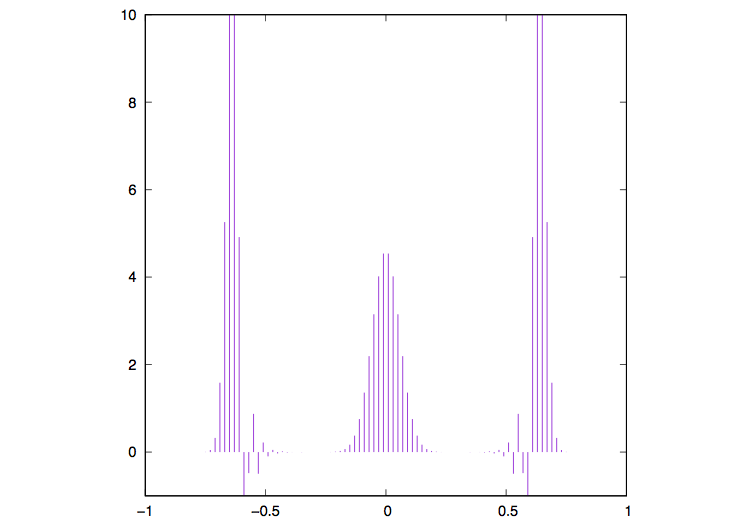}
(c)
\end{minipage}
\caption{Graph of $\mu_{100,M}(x)$ for various $M$}
\label{fig-dynamics_sample_varM_1}
(a): $M=1$.
(b): $M=2$.
(c): $M=3$.
\end{figure}

\begin{figure}[htbp]
\begin{minipage}{0.33\hsize}
\centering
\includegraphics[width=5.5cm]{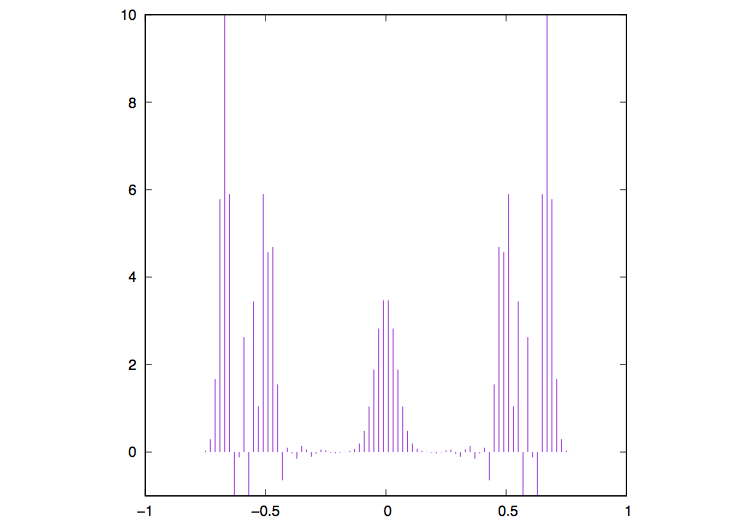}
(a)
\end{minipage}
\begin{minipage}{0.33\hsize}
\centering
\includegraphics[width=5.5cm]{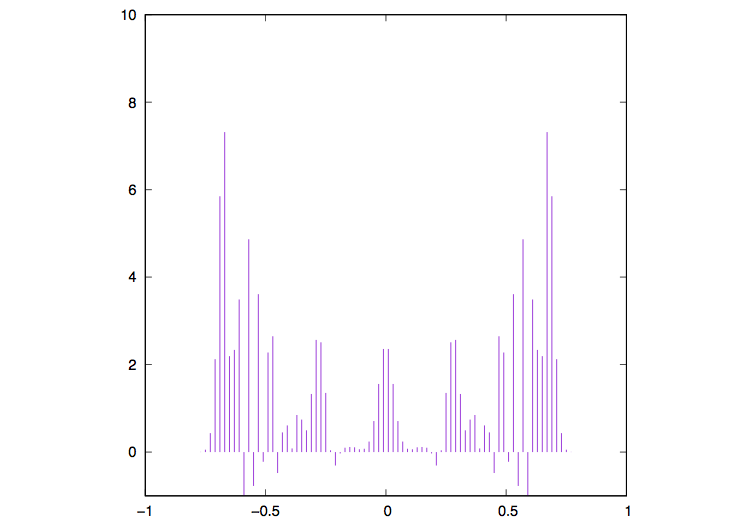}
(b)
\end{minipage}
\begin{minipage}{0.33\hsize}
\centering
\includegraphics[width=5.5cm]{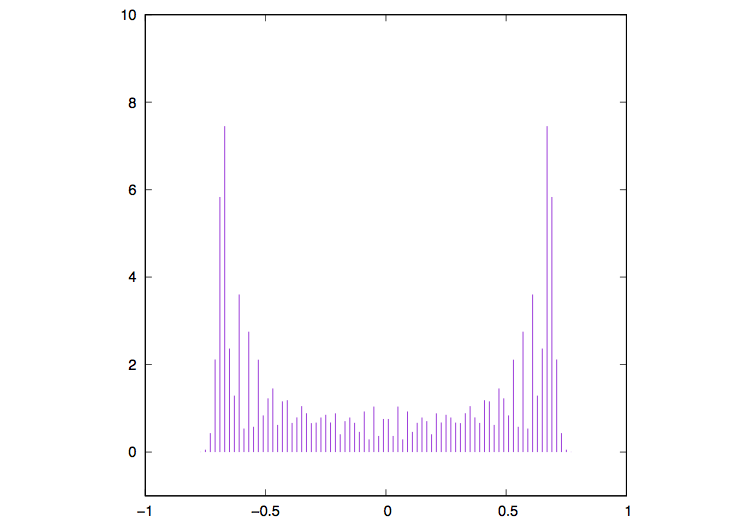}
(c)
\end{minipage}
\caption{Graph of $\mu_{100,M}(x)$ for various $M$: continued}
\label{fig-dynamics_sample_varM_2}
(a): $M=5$.
(b): $M=10$.
(c): $M=51$.
\end{figure}

\begin{figure}[htbp]
\begin{minipage}{0.33\hsize}
\centering
\includegraphics[width=5.5cm]{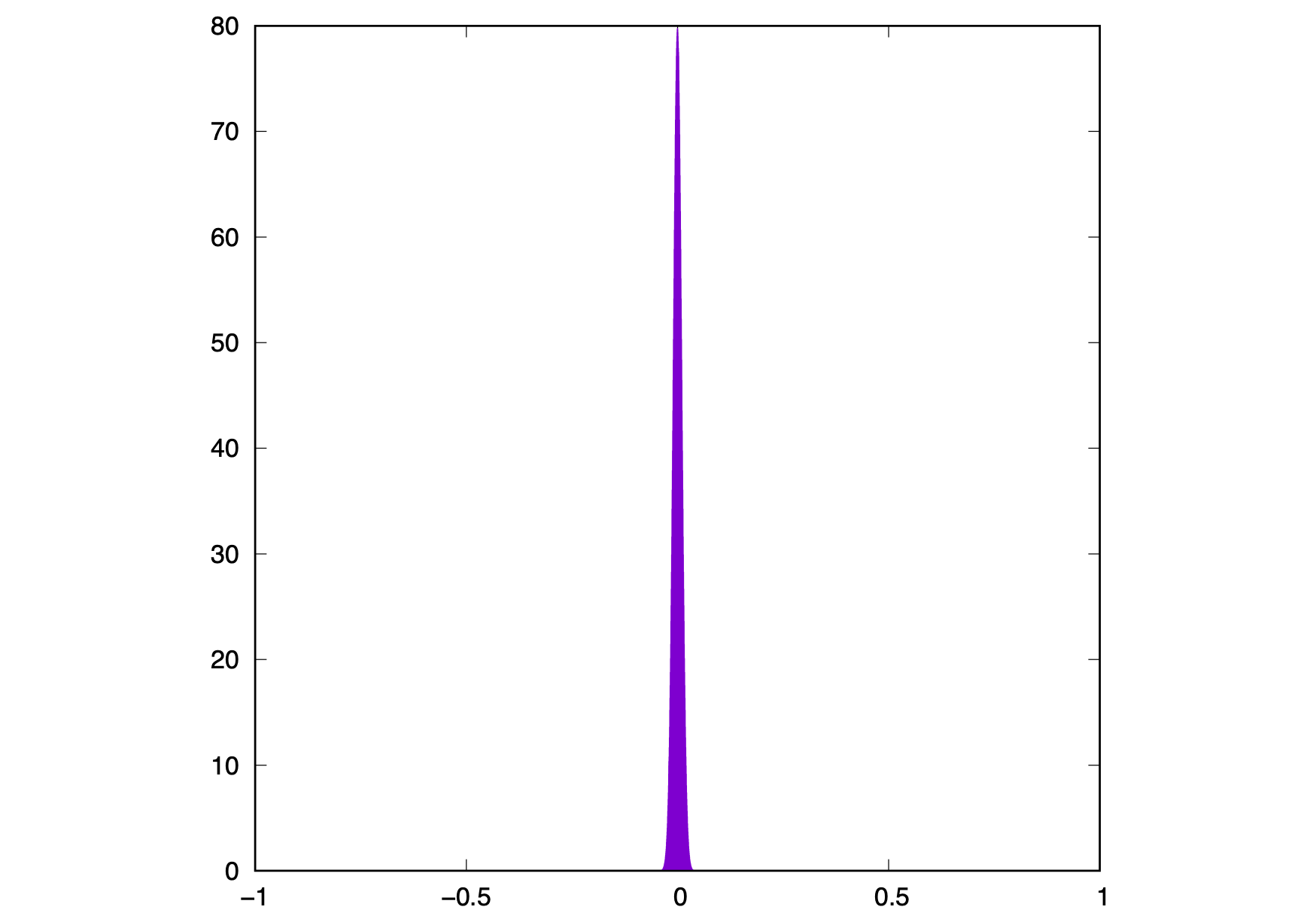}
(a)
\end{minipage}
\begin{minipage}{0.33\hsize}
\centering
\includegraphics[width=5.5cm]{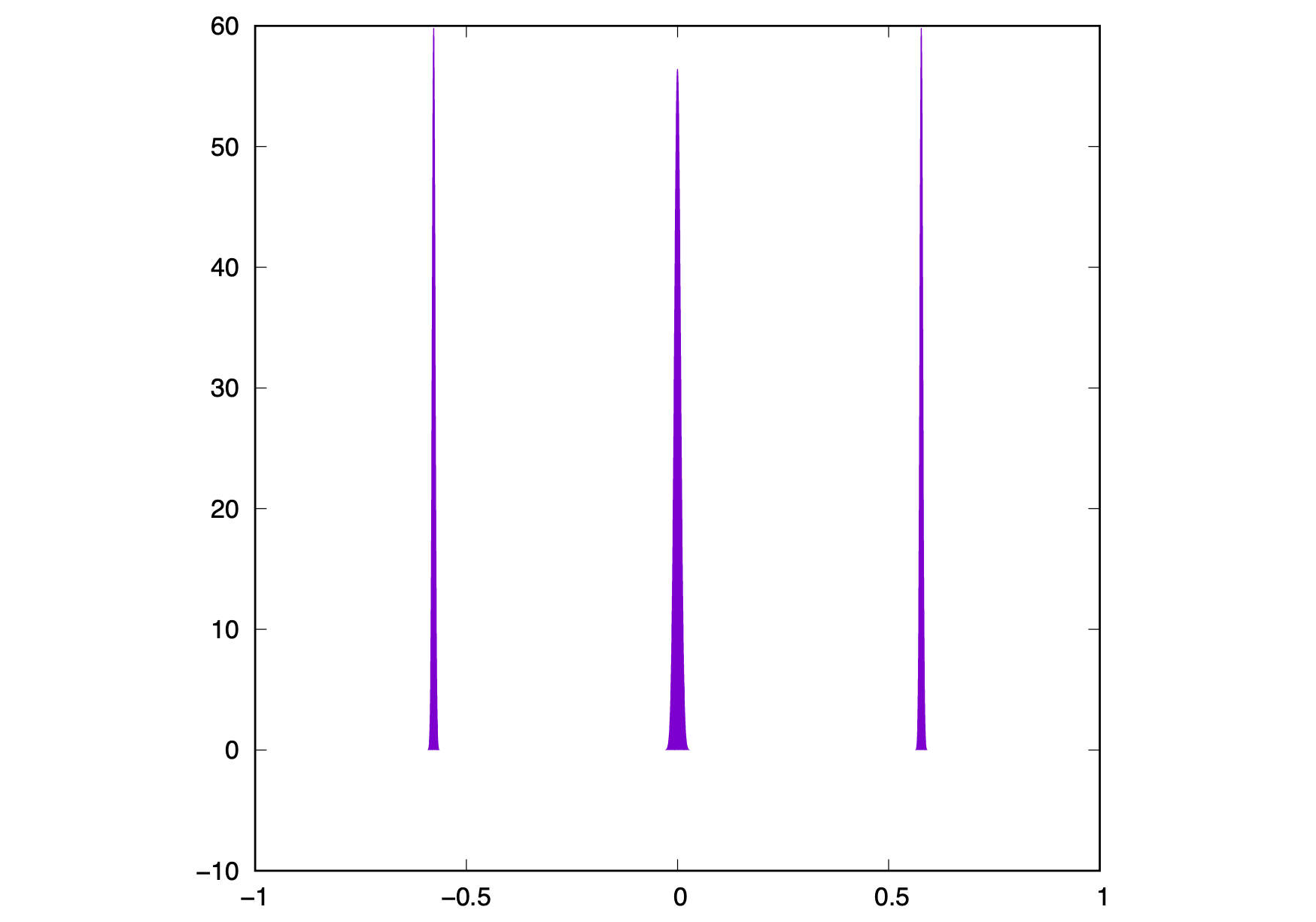}
(b)
\end{minipage}
\begin{minipage}{0.33\hsize}
\centering
\includegraphics[width=5.5cm]{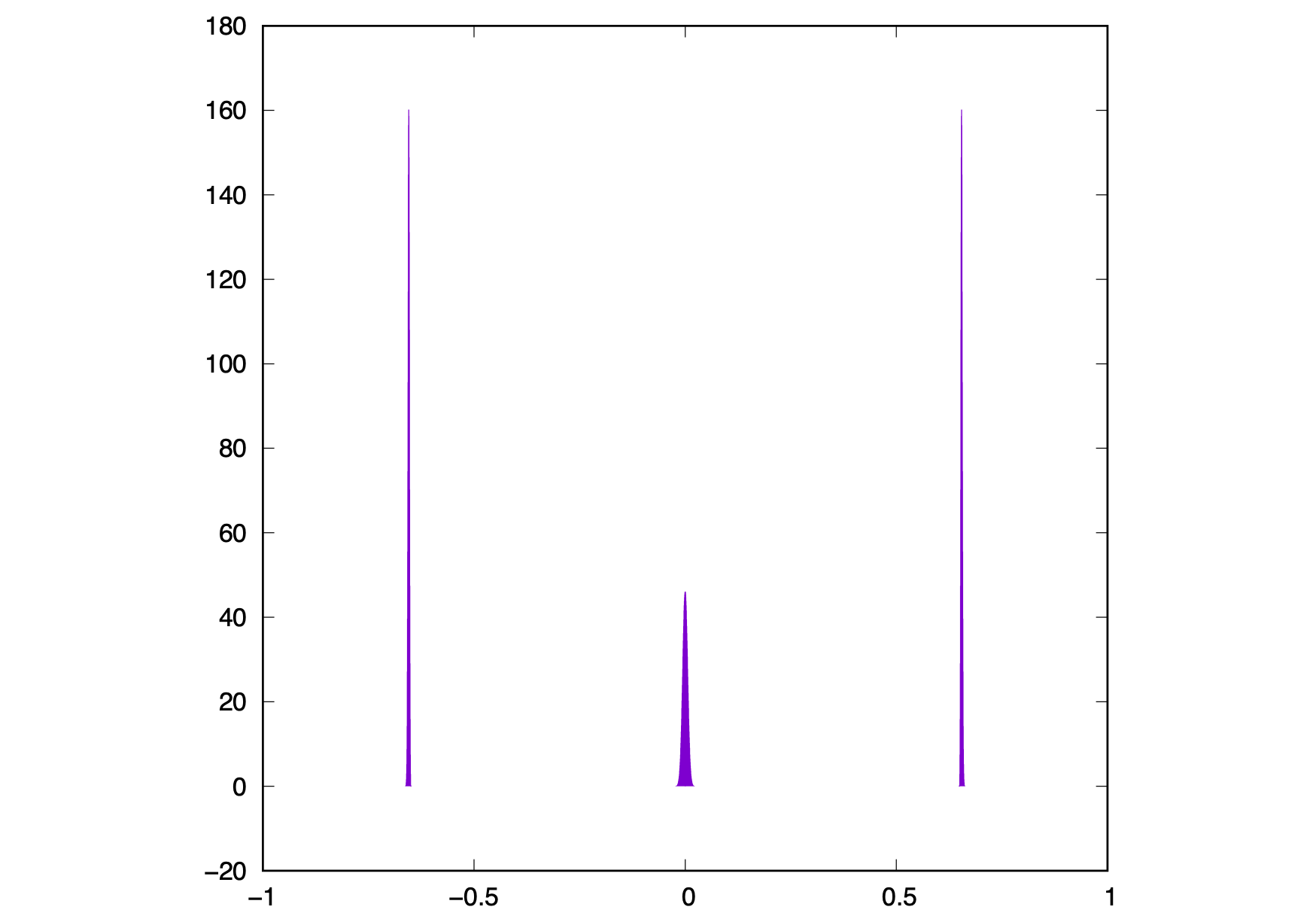}
(c)
\end{minipage}
\caption{Graph of $\mu_{10000,M}(x)$ for various $M$}
\label{fig-dynamics_long_varM_1}
(a): $M=1$.
(b): $M=2$.
(c): $M=3$.
\end{figure}

\begin{figure}[htbp]
\begin{minipage}{0.33\hsize}
\centering
\includegraphics[width=5.5cm]{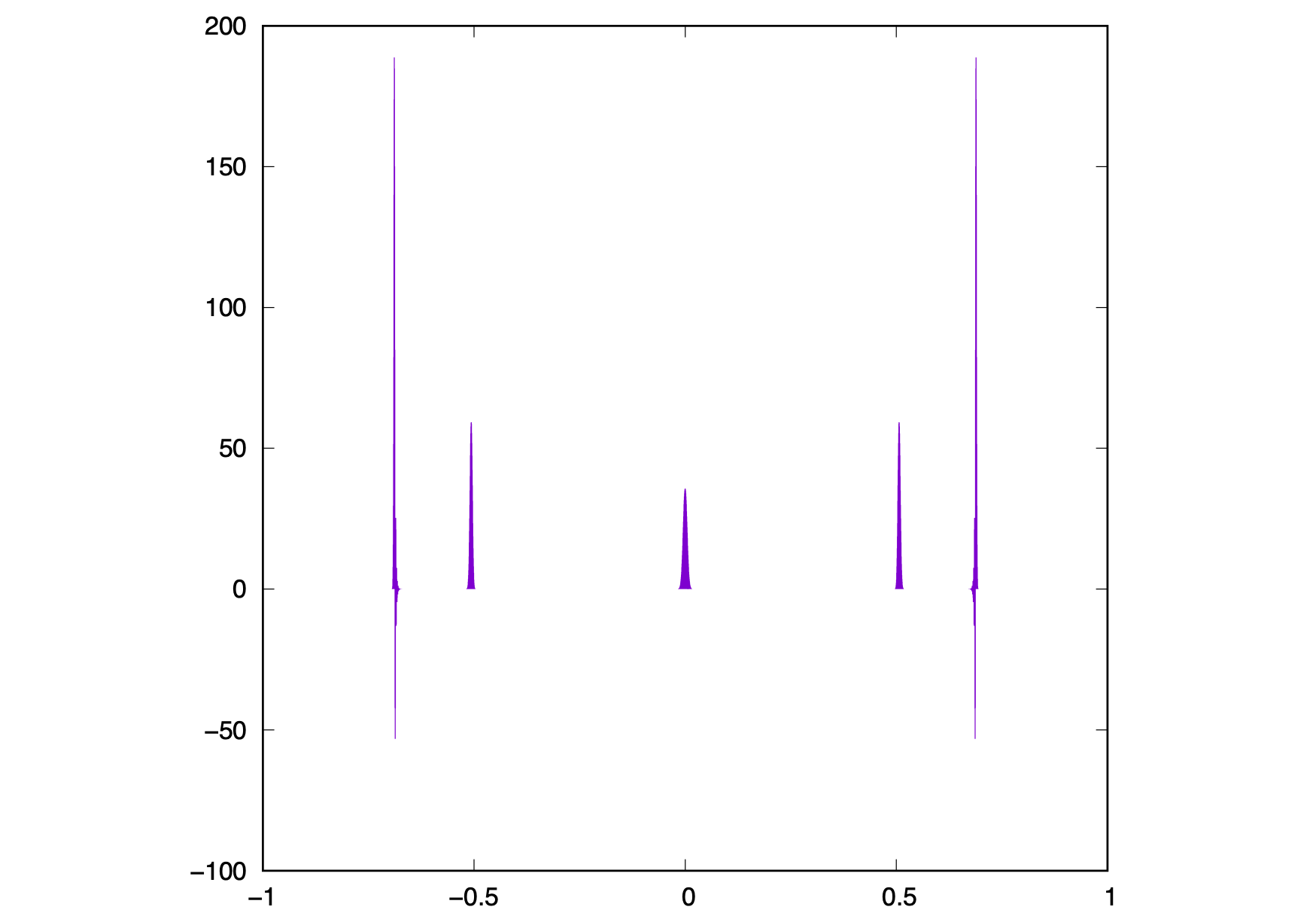}
(a)
\end{minipage}
\begin{minipage}{0.33\hsize}
\centering
\includegraphics[width=5.5cm]{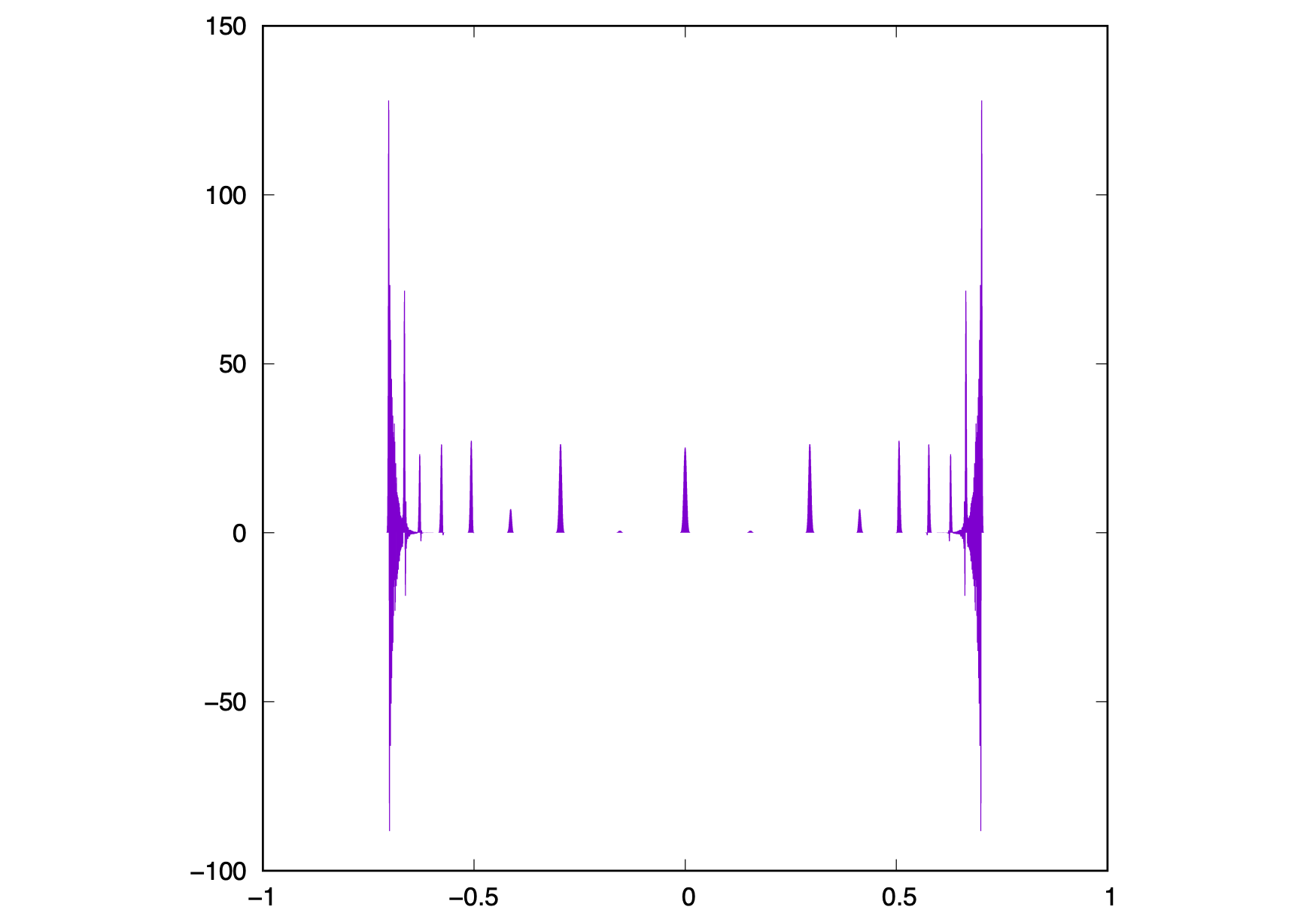}
(b)
\end{minipage}
\begin{minipage}{0.33\hsize}
\centering
\includegraphics[width=5.5cm]{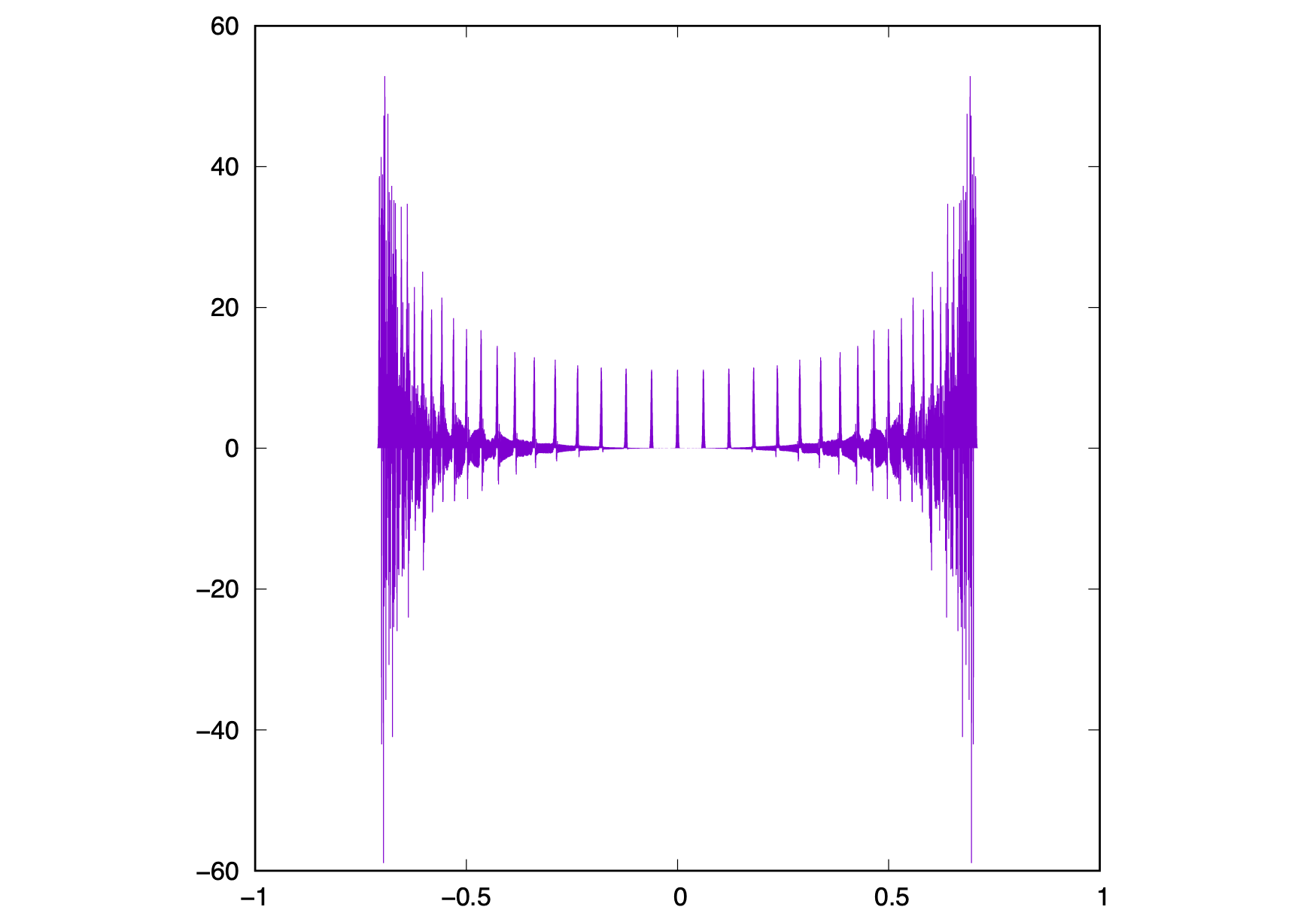}
(c)
\end{minipage}
\caption{Graph of $\mu_{10000,M}(x)$ for various $M$: continued}
\label{fig-dynamics_long_varM_2}
(a): $M=5$.
(b): $M=10$.
(c): $M=51$.
\end{figure}

\subsection{Influence of boundary on dynamics}
\label{section-num-boundary}
Next we study the influence of boundary on walkers.
First we fix $n=100$.
Recall that there is no influence of the boundary on walkers when $M$ is sufficiently large, say $M=201$ for $n=100$ and $M=500$ for $n=200$.
Figure \ref{fig-boundary} shows the walker \lq\lq distributions" with various $M$ in three-dimensional visualizations so that influence of the boundary is easily visible.
If $M$ is relatively large, we see no changes among them, which indicates that the influence of the boundary on the dynamics does not arrive at the distribution on the diagonal.
On the other hand, if $M$ is smaller than $70$, the interference becomes visible at the center when $n=100$.
As $M$ decreases, the regions where $\mu_{n,M}(x,y) > 0$ are localized and several peaks are observed to make a fringe pattern.
As $M$ decreases further, such localized peaks decay and decrease, and finally, the distribution becomes the Gaussian distribution as $M\to 1$.

\begin{figure}[htbp]
\begin{minipage}{0.5\hsize}
\centering
\includegraphics[width=7.5cm]{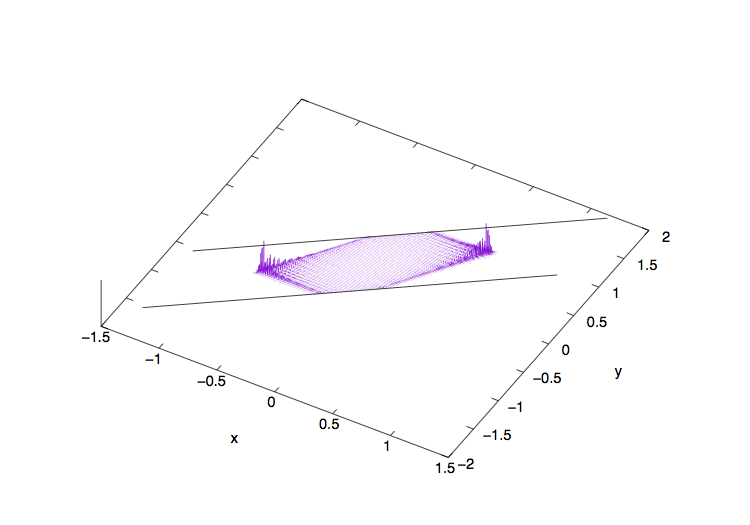}
(a)
\end{minipage}
\begin{minipage}{0.5\hsize}
\centering
\includegraphics[width=7.5cm]{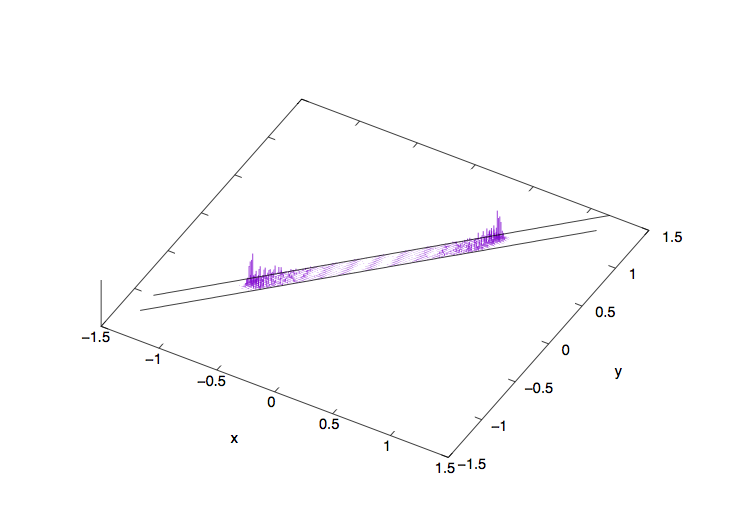}
(b)
\end{minipage}\\
\begin{minipage}{0.33\hsize}
\centering
\includegraphics[width=5.5cm]{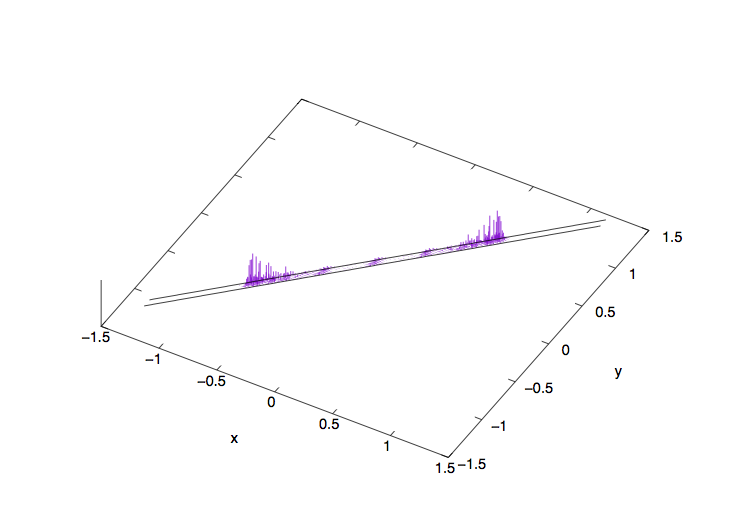}
(c)
\end{minipage}
\begin{minipage}{0.33\hsize}
\centering
\includegraphics[width=5.5cm]{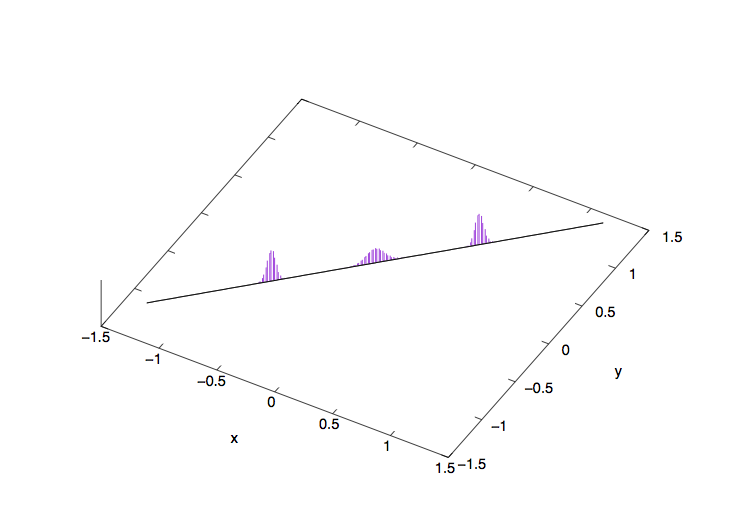}
(d)
\end{minipage}
\begin{minipage}{0.33\hsize}
\centering
\includegraphics[width=5.5cm]{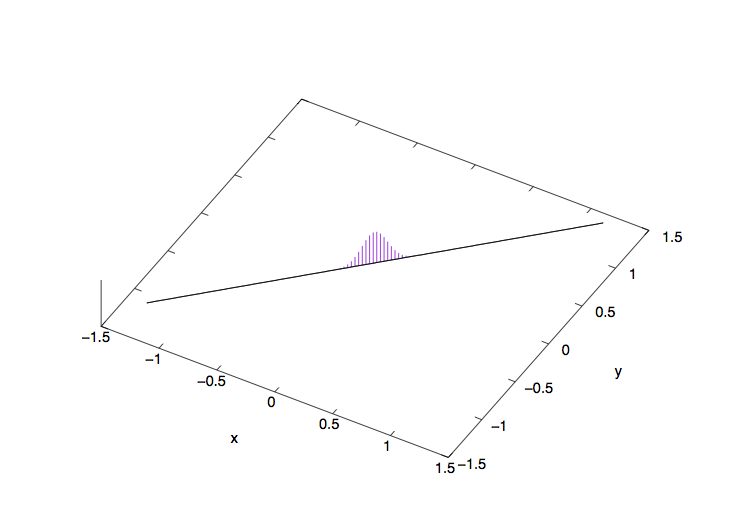}
(e)
\end{minipage}
\caption{Influence of boundary}
\label{fig-boundary}
Black solid lines denote the boundary.
(a): $M=101$. Walkers are affected by the boundary, while the effect is invisible clearly on the diagonal.
(b): $M=21$. 
Interference of the boundary becomes clearly visible, which creates a fringe pattern.
(c): $M=10$. 
Visible fringe patterns become coarser, compared with the case $M=21$.
(d): $M=2$. 
We can observe only three islands. (cf. Figure \ref{fig-dynamics_sample_varM_1}-(b))
(e): $M=1$. 
Walkers accumulate at the center.
In other words, no walkers spread far from the center.
(cf. Figure \ref{fig-dynamics_sample_varM_1}-(a))
\end{figure}

\par
To see the influence of the boundary more precisely, we pay attention to {\em the presence of negative distributions}. 
In Figure \ref{fig-dynamics_sample}, we cannot see positions where $\mu_{n,M}(x)$ becomes negative, in which case $M$ is relatively large so that the walker does not arrive at the boundary.
On the other hand, in the case that $M$ is relatively small compared with time step $n$, we see negative value distributions of $\mu_{n,M}(x)$ separating \lq\lq islands" of distributions attaining positive values (see Figures \ref{fig-dynamics_sample_varM_1}-\ref{fig-dynamics_long_varM_2}), which make fringe patterns in Figure \ref{fig-boundary}.
We then see the presence of negative values in $\{\mu_{n,M}(x)\}$ as the onset of the influence of the boundary.
To see the tendency precisely, we define the critical time $n_{crit} = n_{crit}(M)$ as follows:
\begin{equation}
\label{n_crit}
n_{crit}(M) := \sup\{n \mid \mu_{\tilde n,M}(x) \geq 0 \text{ for all $x\in \mathbb{Z}$ and ${\tilde n}\in \{0,1,\cdots, n\}$} \}.
\end{equation}
Figure \ref{fig-critical-N} shows the graphs of $n_{crit}(M)$, which indicates that $n_{crit}(M)$ is almost proportional to $M$ whose slope depends on whether $M$ is even or odd. 
More precisely, for $M \leq 100$, we have
\begin{equation*}
n_{crit}(M) \approx
\begin{cases}
3M & \text{$M$ is even},\\
2M & \text{$M$ is odd}.
\end{cases}
\end{equation*}
As $M$ increases, the slopes are gradually decreased, while the onset of negative distributions within the time $O(M)$ is qualitatively unchanged.
We say here that $n_{crit}(M) = O(M)$ can characterize the influence of boundary on dynamics of walkers and qualitative change of their limit distributions.

\begin{figure}[htbp]
\centering
\includegraphics[width=7.5cm]{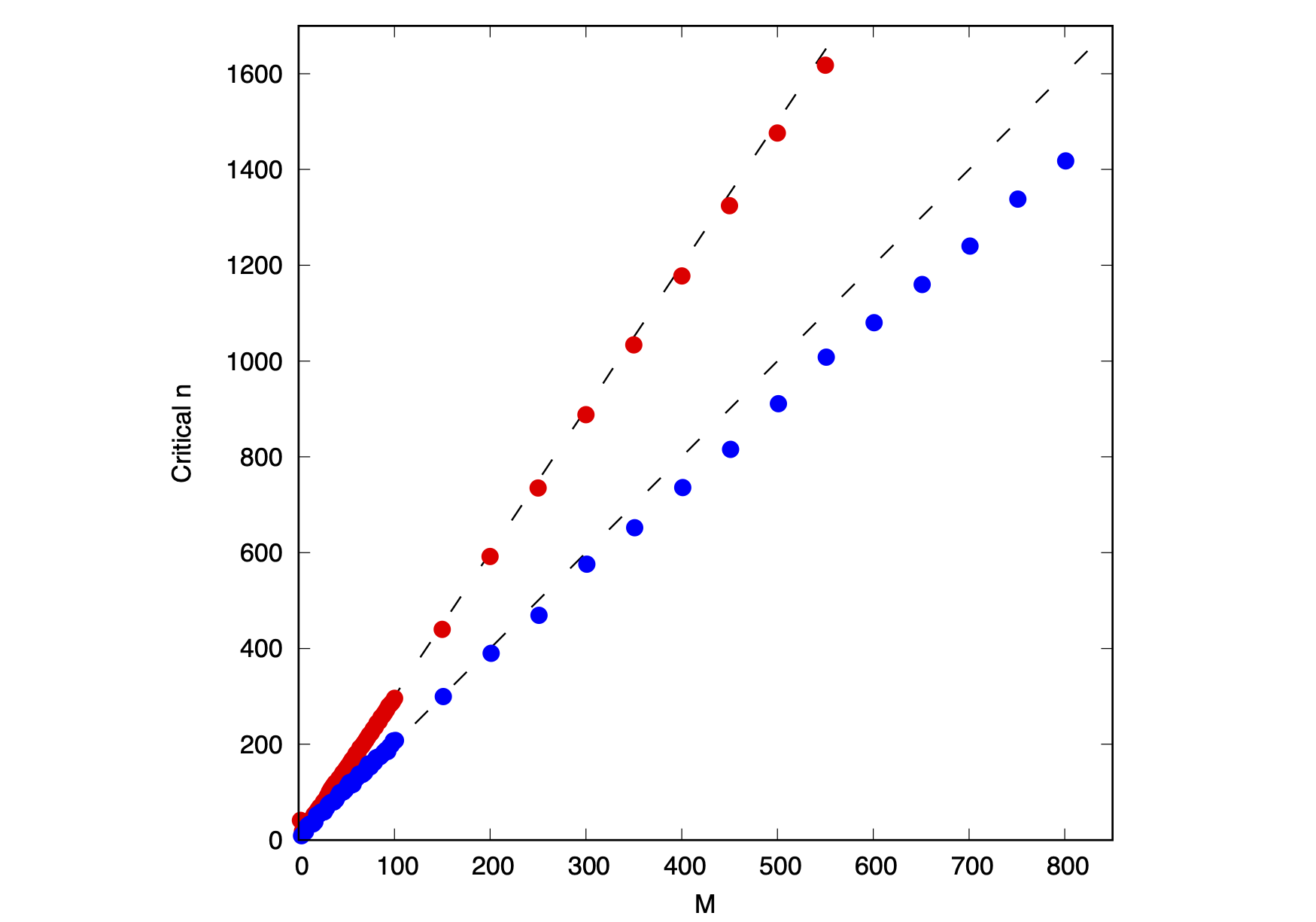}
\caption{Graph of $n_{crit}(M)$ in (\ref{n_crit})}
\label{fig-critical-N}
Horizontal: $M$. Vertical: $n_{crit}(M)$.
Red: plot of $n_{crit}$ for even $M$. 
Namely $n_{crit}(M)$ as a function of even $M$. 
Blue: plot of $n_{crit}$ for odd $M$ in the similar manner to the red data. 
Dotted lines show $n_{crit} = 2M$ and $n_{crit} = 3M$.
\end{figure}

\subsection{Characteristics of localized peaks}
\label{section-num-localized}
As $M$ increases, the limiting behavior of walker's distribution looks including several nature of the Hadamard-type quantum walk.
Here we study such characteristics of distributions to see how close the limiting behavior of walkers is to the Hadamard walk.
There are several well-known characteristics of asymptotic behavior in quantum walks such as the Hadamard's one (e.g. \cite{Sunada_Tate2012}).
We then study the following values as a function of $n$:
\begin{itemize}
\item the normalized position $\bar x^{\max}_{n, M}$ where the density attains the maximal value,
\item ratio of the height of $\mu_{n,M}(0)$ to $\mu_{n,M}(\bar x^{\max}_{n, M})$,
\item asymptotic size of tail outside the limiting support $d_n$,
\item decay rate of $\mu_{n,M}(0)$ as $n\to \infty$,
\item decay rate of $\mu_{n,M}(\bar x^{\max}_{n, M})$ as $n\to \infty$,
\end{itemize}
whose precise definitions are shown below.
Note that all these values can depend on $M$.
These values can {\em quantitatively} characterize the closeness of the limiting behavior of the present model to well-known quantum walks.
In all computation results we show below, we have fixed $n = 5000$ and note that the corresponding values are almost the same between $5000$ and $10000$.
All computed data discussed below are summarized in Tables \ref{Table-ch1} and \ref{Table-ch2}.

\subsubsection{The position of the density attaining the maximal value}
Looking at all figures in the previous subsections, we see that the leftmost and rightmost peaks of distributions of walkers rarely change as $M$ increases.
The present step is to study the peak position more precisely.
To this end, we extract the position of peaks and its time evolution.
We have computed the normalized peak position of walker distributions outside the center distribution by determining
\begin{equation}
\label{position}
\bar x^{\max}_{n, M} :=\left\{ \bar x \in [\delta,1] \mid \mu_{n,M}(\bar x) = \max_{\bar x\in [\delta,1]} \mu_{n,M}(x)\right\}.
\end{equation}
As a sample result, we have calculated $\left\{\bar x^{\max}_{5000, M} \right\}$ for various $M$. 
Here $\delta > 0$ is a number sufficiently small so that the centered random-walk-like distribution is concentrated on the interval $(-\delta, \delta)$.
Note that the sequence $\left\{\bar x^{\max}_{n, M}\right\}_{n\geq 1}$ is also parameterized by the width $M$ of the stripe.
The sequence $\left\{\bar x^{\max}_{n, 2}\right\}_{n\geq 1}$ converges to $1/\sqrt{3}$, while $\left\{\bar x^{\max}_{n, 3}\right\}_{n\geq 1}$ to $0.64$.
Approximate values of $\bar x^{\max}_{n, M}$ during our computations are summarized in Table \ref{Table-ch1}.
As $M$ increases, the sequence $\left\{\bar x^{\max}_{N, M}\right\}_{M\geq 1}$ becomes distributed around $0.69\sim 0.71$.
This observation shows that the walker spreads depending linearly on $n$ like quantum walks, while the rate of linearity has an influence on the boundary. 
More precisely, the thinner the stripe is, the smaller the rate is.
It is also noted that our numerical result for $M=2$ follows our mathematical argument, Theorem \ref{theorem-limit-behavior}. 

\subsubsection{Ratio of the height of $\mu_{n,M}(0)$ to $\mu_{n,M}(\bar x^{\max}_{n, M})$}
\label{section-ratio}
Next we study the ratio of the height of $\mu_{n,M}(0)$ to $\mu_{n,M}(\bar x^{\max}_{n, M})$.
As we see in Figure \ref{fig-dynamics_long_varM_1}, the height of the distributions among the center and the side looks almost identical for $M=2$, whereas it becomes significantly different for $M\geq 3$.
We have calculated the ratio $\mu_{n,M}(0) / \mu_{n,M}(\bar x^{\max}_{n, M})$ for various $M$.
In practical calculations, we have calculated the following average:
\begin{equation}
\label{average-ex}
\frac{1}{3000}\sum_{n=2001}^{5000}\frac{\mu_{n,M}(0)}{\mu_{n,M}(\bar x^{\max}_{n, M})}
\end{equation}
as the ratio we are interested in.
Computation results are summarized in the row \lq\lq ratio" of Tables \ref{Table-ch1} and \ref{Table-ch2}.
As indicated before, the ratio for $M=2$ is about $0.4714$, whereas the ratio drastically changes in $M\geq 3$ and we also see that the ratio stays at the range between $0.18$ and $0.2$ as $M$ increases.
\par
\bigskip
It should be noted here that there is a trick in the difference of the ratio $\approx 0.47$ for $M=2$ from visual observation in Figure \ref{fig-dynamics_long_varM_1}-(b). 
By definition of our model, we know that $\mu_{n,M}(0) = 0$ for {\em all odd} $n$.
On the other hand, the maximal point $\bar x^{\max}_{n, M}$ of $\mu_{n,M}$ changes depending on time $n$.
Taking an average like (\ref{average-ex}), the density at the origin $\mu_{n,M}(0)$ contributes only at even $n$. 
The genuine average ratio (\ref{average-ex}) is therefore a half of the height ratio of densities between the center and sides. 

\subsubsection{Asymptotic size of tail outside the limiting support}

Next we further study the nature of localized peaks.
A well-known result is that random walks followed by the normal distribution $N(0,1)$ (the mean $0$ and the variance $1$) has the probability distributions of the width $O(\sqrt{n})$, where $n$ is the number of iterations, while quantum walks typically have peak distribution of the width $O(n^{1/3})$ (cf. \cite{Sunada_Tate2012}).
To this end, we define the width $d_n$ of the tail distribution outside the peak of density as follows:
\begin{equation*}
d_n := a_n - \bar x^{\max}_{n, M},\quad \text{ where }\quad {\rm supp}\, \mu_{n,M} = [-a_n, a_n].
\end{equation*}
Assuming that $d_n$ follows the asymptotic behavior
\begin{equation}
\label{gamma}
d_n \sim O(n^{\gamma_M})\quad \text{ as }\quad n\to \infty,
\end{equation}
we calculate $\gamma_M$.
Computed results of $\gamma_M$ for various $M$ are shown in Tables \ref{Table-ch1} and \ref{Table-ch2}.
We see that $\gamma_M$ converges to about $0.345\approx 1/3$ as $M$ increases, which indicates that the limiting behavior of the tail becomes close to well-known quantum walks according to e.g. \cite{Sunada_Tate2012}.

\begin{remark}
While computations of $\gamma_M$, we have chosen localized region of $n$ where $\gamma_M$ is almost constant.
In fact, there is a case that $d_n$ frequently oscillates during evolution to prevent us from computing $\gamma_M$, which occurs when $M=21, 51, 75$ in our computations.
Nevertheless the order exponent $\gamma_M$ is observed to be almost identical among different localized regions (in the above sense), so we have chosen one of such regions to compute $\gamma_M$.
When $M=2,3,5,10,101, 126$, $\gamma_M$ has achived to be constant over $n\in [2000,5000]$.
\end{remark}

\subsubsection{Decay rate of $\mu_{n,M}(0)$ and $\mu_{n,M}(\bar x^{\max}_{n, M})$ as $n\to \infty$}
Finally we compute asymptotic decay rate of density $\mu_{n,M}$ inside the support.
More precisely, we assume the following asymptotic behavior
\begin{equation}
\label{decay-walker}
\mu_{n,M}(x) \sim O(n^{r_M(x)})\quad \text{ as }\quad n\to \infty
\end{equation}
and compute $r_M(x)$.
As representatives, we study the behavior of $\mu_{n,M}(0)$ and $\mu_{n,M}(\bar x^{\max}_{n, M})$.
Correspondingly we define
\begin{equation*}
r_M^{{\rm center}} := r_M(0),\quad r_M^{{\rm side}} := r_M(\bar x^{\max}_{n, M}).
\end{equation*}
Note that, as mentioned in calculations of (\ref{average-ex}), $\mu_{n,M}(0)$ is always $0$ for odd $n$.
Nevertheless this restriction does no contribution to $r_M(0)$.
Computation results are shown in Table \ref{Table-ch1}.
These results show that $r_M^{{\rm center}}$ becomes constant $-1/2$, while $r_M^{{\rm side}}$ depends on $M$ in a nonlinear manner and converges to $-0.6564 \approx -2/3$ as $M$ increases.

\begin{table}[ht]
\caption{Characteristics of off-diagonal peaks}
\centering
\begin{tabular}{c|ccccc}
\hline
$M$ & $2$ & $3$ & $5$ & $10$ & $21$  \\ [2mm]
\hline
$\bar x^{\max}_{M}$ & $0.6917$ & $0.6935$ & $0.6931$ & $0.7046$ & $0.6963$ \\ [2mm]
\hline
ratio & $0.4717$ & $0.1441$ & $0.0960$ & $0.1014$ & $0.0996$ \\ [2mm]
\hline
$\gamma_M$ & $0.4959$ & $0.4388$ & $0.3713$ & $0.3489$ & $0.4101$  \\[2mm]
\hline
$r_M^{{\rm center}}$ & $-0.4999$ & $-0.4998$ & $-0.4996$ & $-0.4992$ & $-0.4983$  \\ [2mm]
\hline
$r_M^{{\rm side}}$ & $-0.4990$ & $-0.4830$ & $-0.3996$ & $-0.3368$ & $-0.4367$ \\ [2mm]
\hline
\end{tabular}%
\label{Table-ch1}\par
$\bar x^{\max}_{M} = \bar x^{\max}_{5000,M}$ is given in (\ref{position}). 
\lq\lq ratio" denotes $\mu_{n,M}(0) / \mu_{n,M}(\bar x^{\max}_{n, M})$ discussed in Section \ref{section-ratio}.
$\gamma_M$ is given in (\ref{gamma}).
Finally, $r_M^{{\rm center}}$ and $r_M^{{\rm side}}$ are characterized by (\ref{decay-walker}).
\end{table}

\begin{table}[ht]
\caption{Characteristics of off-diagonal peaks: continued}
\centering
\begin{tabular}{c|cccc}
\hline
$M$ & $51$ & $75$ & $101$ & $126$ \\ [2mm]
\hline
$\bar x^{\max}_{M}$ & $0.6917$ & $0.6935$ & $0.6931$ & $0.7046$ \\ [2mm]
\hline
ratio & $0.0940$ & $0.0949$ & $0.0922$ & $0.0757$ \\ [2mm]
\hline
$\gamma_M$ & $0.4891$ & $0.5048$ & $0.3450$ & $0.3449$ \\[2mm]
\hline
$r_M^{{\rm center}}$ & $-0.4959$ & $-0.4939$ & $-0.4919$ & $-0.4898$ \\ [2mm]
\hline
$r_M^{{\rm side}}$ & $-0.5811$ & $-0.6248$ & $-0.6563$ & $-0.6564$\\ [2mm]
\hline
\end{tabular}%
\label{Table-ch2}\par
Expressions are the same as Table \ref{Table-ch1}.
\end{table}

\subsection{Comparison with analytic results}
\label{section-num-comparison}
Here we compare our numerical results with analytic results from two aspects.
The first one is the comparison of limiting behavior among central limits in open quantum random walks and our present model, in particular (\ref{eq:clt+})-(\ref{eq:clt-}).
Another is the description of eigenvectors generating specific behavior under quantum walk dynamics.

\subsubsection{Limiting behavior}
In $M=2$, we have several analytic results for the asymptotic behavior of walkers.
First note that the normalized position $\bar x_{n,M}^{\max}$ of the side peak corresponds to the spreading speed of the ballistic mode (in positive direction), which is $1/\sqrt{3} \approx 0.5774$, according to Theorem \ref{theorem-limit-behavior}.
The corresponding numerical result shows the perfect agreement with the above analytic results; $\bar x_2^{\max}$ in Table \ref{Table-ch1}.
\par
As for the ratio of the height, we apply the fact that the function $e^{-a^2 k^2/2}$ is the Fourier transform of
\begin{equation*}
f(y; a^2) = e^{-y^2/(2a^2)} / \sqrt{2 \pi a^2},
\end{equation*}
and $a^2 = 1/2$ for the center (Lemma \ref{lemma-lambda-1}) and $a^2 = 4/9$ for the sides (Lemma \ref{lemma-lambda-2}). 
The actual value is $f(0; 1/2) = \sqrt{1/\pi}$ and $f(0; 4/9) = \sqrt{9/8\pi}$, respectively.
However, $\mu_{n,M}(x) = 0$ for all odd $n$ if $x$ is even, and hence we cannot compare the height of the density distribution to obtain useful information of limiting behavior.
In practice, the sum of $\mu_{n,M}$ in the center distribution factors $c_0 = 1/2$ into the limiting behavior $\int_a^b f(y; a^2)dy$ in (\ref{eq:clt+}), which reflects the fact that $\mu_{n,M}(x)$ is nonzero only at discretely distributed $x$, while the (localized) distribution looks like the classical random walk-like one.
The similar situation occurs in the side distribution (\ref{eq:clt-}).
Consequently, our numerical results in $M=2$ follows mathematical arguments and, potentially in $M\geq 3$, shows qualitatively and quantitatively intrinsic nature of our quantum walk model.

\subsubsection{Eigenvectors}

The normalized limiting behavior with $M=2$ consists of three pieces; stationary components at the center and spreading components in both sides.
In Theorem \ref{thm:eigenprojection} (and Proposition \ref{prop:eigensystem}) we have calculated eigenvectors $v_1, v_2$ and $v_3$ of $\hat W_{s,t}(\delta)$ associated with eigenvalues $1-\delta^2 / 4$ and $1\pm \frac{i}{\sqrt{3}}\delta - \frac{2}{9}\delta^2$, respectively, up to $O(\delta^2)$-corrections.
We can confirm that these eigenvectors indeed generates the mentioned behavior in an integral sense.
The corresponding behavior of walkers with the initial data $v_1, v_2$ and $v_3$ are drawn in Figure \ref{fig-eigen}.
Each walker possesses the genuinely positive density distribution as well as extra ones equally distributed in both positive and negative signs.
For example, time evolution of $v_1$, Figure \ref{fig-eigen}-(a), possesses the density distribution with positive integral at the center, while there are also densities with positive and negative signs in both sides of the center distribution.
We have numerically confirmed that the sum $\sum_{x} \mu_{n,M}(x)$ of densities over localized regions is $0$ for all $n$, as far as the side localized regions are clearly distinguished, in which sense $v_1$ generates the stationary components of the present quantum walk.
The similar behavior are observed for $v_2$ and $v_3$, which generate moving component to the right (Figure \ref{fig-eigen}-(b)) and to the left (Figure \ref{fig-eigen}-(c)), respectively. 
The remaining components do exist but the sum of densities over localized regions are $0$ for all large $n$ in the above sense.

\begin{figure}[htbp]
\begin{minipage}{0.33\hsize}
\centering
\includegraphics[width=6cm]{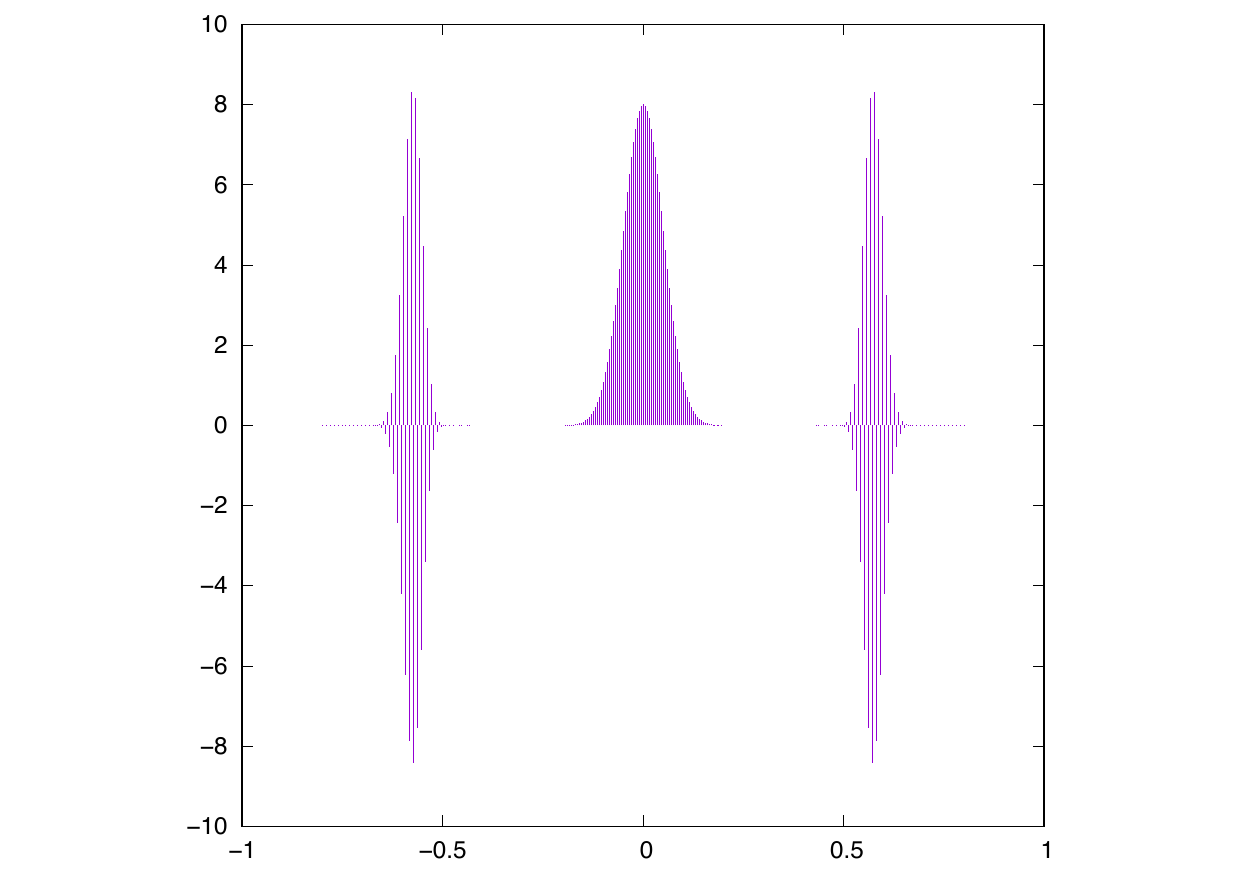}
(a)
\end{minipage}
\begin{minipage}{0.33\hsize}
\centering
\includegraphics[width=6cm]{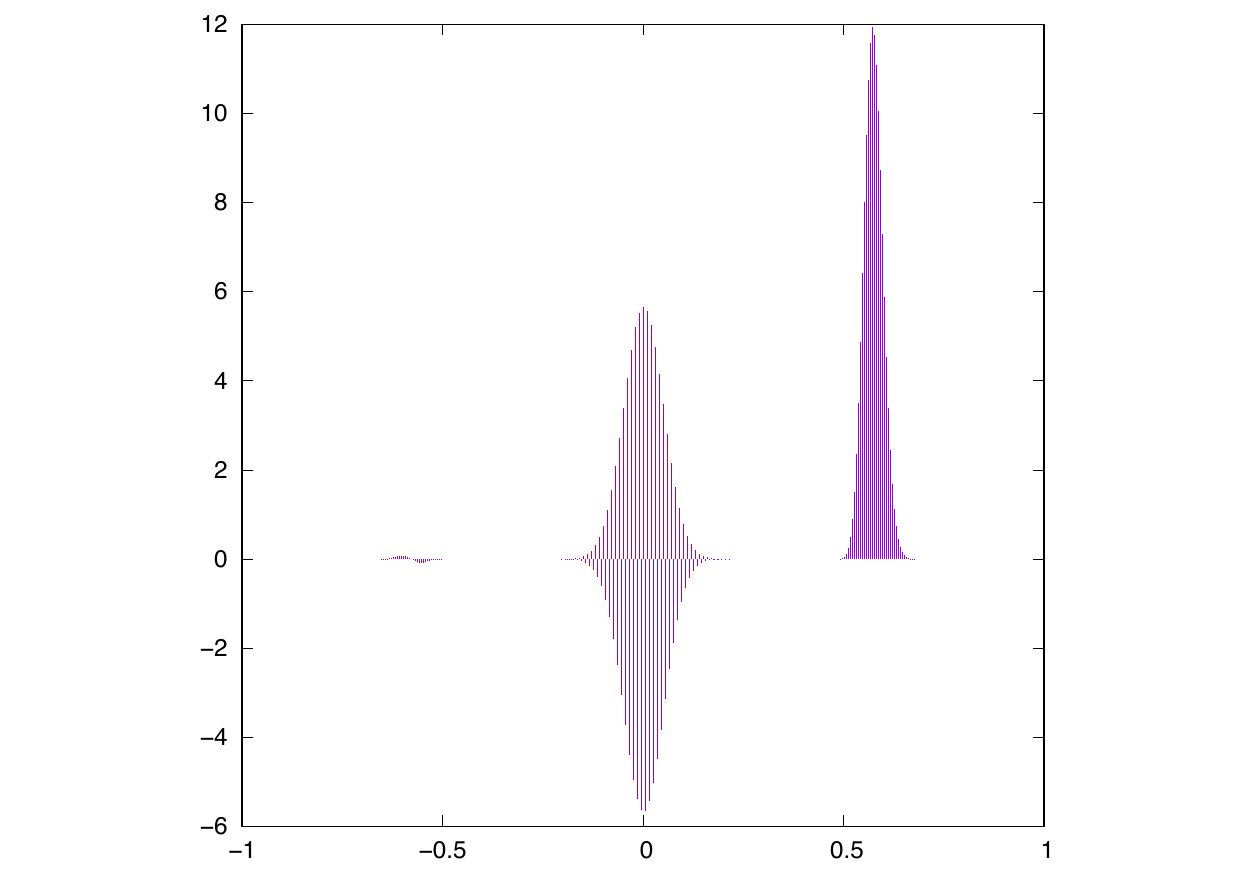}
(b)
\end{minipage}
\begin{minipage}{0.33\hsize}
\centering
\includegraphics[width=6cm]{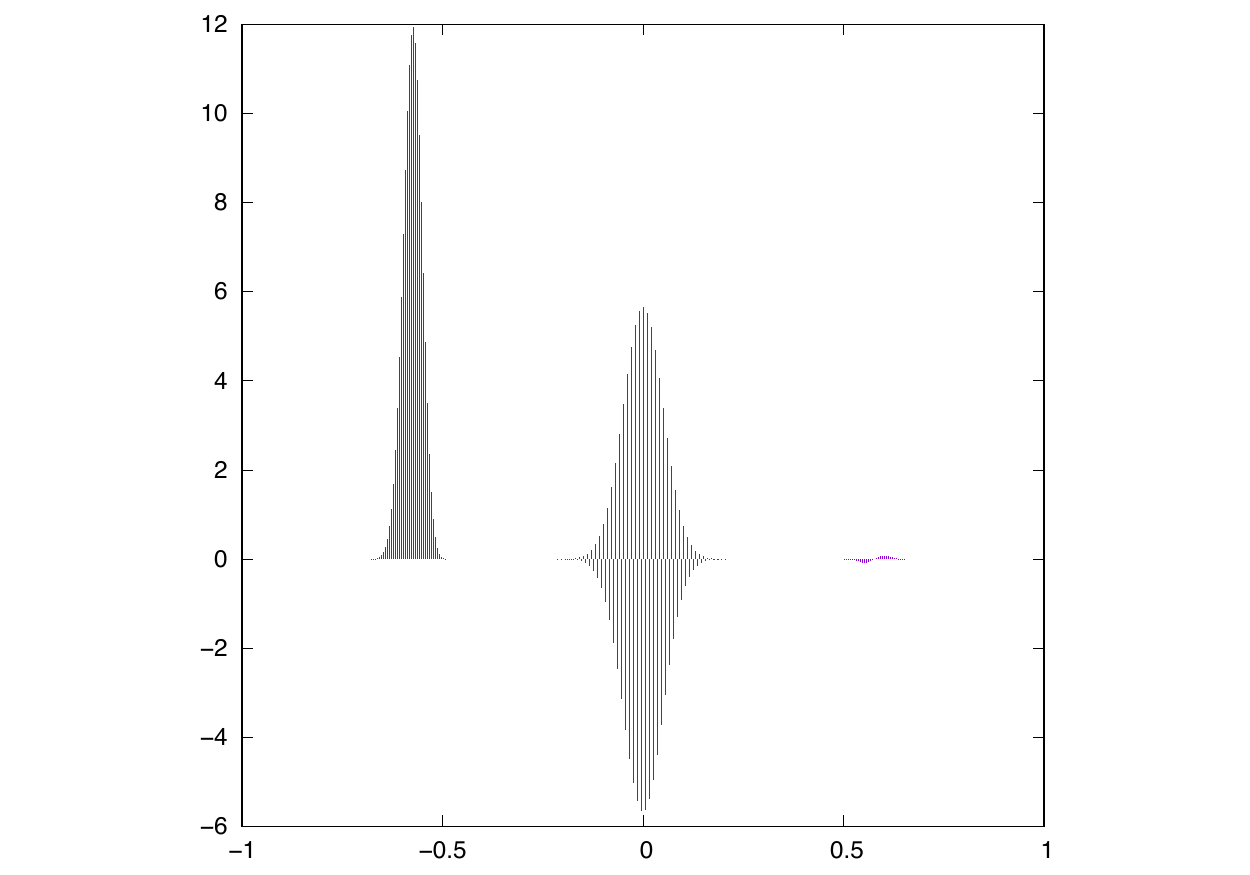}
(c)
\end{minipage}
\caption{Behavior of eigenvectors}
\label{fig-eigen}
All figures show $200$-step evolution of under our quantum walk model with $M=2$ initialized at the following eigenvectors: (a): $v_1$, (b): $v_2$, (c): $v_3$.
These eigenvectors generate stationary component, spreading components to the right and to the left, respectively, in the sense that the remaining components has the sum $0$ over the corresponding localized regions.
\end{figure}

%% file: eigenspace_unperturbed.tex
To understand the qualitative nature of eigenspaces for $\hat{W}_{s,t}(k)$, we study eigenspaces of  $\hat{W}_{s,t}:=\hat{W}_{s,t}(0)$ and their perturbations $\hat{W}_{s,t}(k)$ for small $k$ through the perturbation theory of linear operators (e.g. \cite{Kato1982}).
The representation matrix for $\hat{W}_{s,t}(0)=:\hat{W}_{s,t}$ is explicitly written as 
\begin{equation*}
\hat{W}_{s,t}= \frac{1}{2}
\left[
\begin{array}{cccc|cccc}
1 & 0 & 0 & 1 & 0 & 1 & 0 & 0  \\
1 & 0 & 0 & -1 & 0 & -1 & 0 & 0 \\
1 & 0 & 0 & -1 & 0 & 1 & 0 & 0 \\
1 & 0 & 0 & 1 & 0 & -1& 0 & 0 \\ \hline
0 & 0 & 1 & 0 & 1 & 0 & 0 & 1 \\
0 & 0 & 1 & 0 & 1 & 0 & 0 & -1 \\
0 & 0 & -1 & 0 & 1 & 0 & 0 & -1 \\
0 & 0 & -1 & 0 & 1 & 0 & 0 & 1
\end{array}
\right].
\end{equation*}
Using the formula for $\lambda_1(k)$ and $\lambda_2(k)$, we know that eigenvalues of $\hat{W}_{s,t}$ are 
\begin{equation*}
\left\{0,0,1,1,1,-\frac{1}{2}, \frac{-1\pm i}{4}\right\},
\end{equation*}
which shows that the matrix $\hat{W}_{s,t}$ possesses eigenvalues with multiplicity greater than $1$. 
By the Cayley-Hamilton theorem, we know that
$F_{\hat{W}_{s,t}}(\hat{W}_{s,t})= O$, where 
\begin{equation}\label{eq:eigeneq}
F_{\hat{W}_{s,t}}(\lambda) = \lambda^2 (\lambda - 1)^3 (2\lambda^2 + \lambda + 1)(2\lambda+1).
\end{equation}
Interestingly, we observe the following result.
\begin{lemma}
The minimal polynomial of the matrix $f_{\mathcal{W}_{s,t}(0)}(\lambda)$ is 
\begin{equation*}
f_{\hat{W}_{s,t}}(\lambda) = \lambda (\lambda - 1) (2\lambda^2 + \lambda + 1)(2\lambda+1).
\end{equation*}
\end{lemma}
\begin{proof}
Direct calculations yield that
\begin{equation*}
\hat{W}_{s,t} (\hat{W}_{s,t} - I) (2\hat{W}_{s,t}^2 + \hat{W}_{s,t} + I)(2\hat{W}_{s,t}+I) = O.
\end{equation*}
\end{proof}
A direct consequence of this lemma is the following.
\begin{proposition}
\label{prop-semisimple}
All eigenvalues of $\hat{W}_{s,t}$ are semisimple. 
In particular, the matrix $\hat{W}_{s,t}$ is diagonalizable.
\end{proposition}
We are interested in eigenspaces of $\hat{W}_{s,t}$ associated with $\lambda = 1 + O(\epsilon)$, since initial data in the corresponding eigenspaces characterize the intrinsic asymptotic behavior of the present quantum walks.
Thanks to Proposition \ref{prop-semisimple}, the eigenvalue $\lambda = 1$ of $\hat{W}_{s,t}$ is semisimple and hence the associating eigenspace $E_1$, which is a three-dimensional linear space, is generated by three linearly independent eigenvectors of $\lambda = 1$.
We have the following observation for eigenstructures.
\begin{proposition}\label{prop:eigenvec}
The eigenspace of $\hat{W}_{s,t}$ associated with $\lambda = 1$ is generated by the following vectors:
\begin{equation*}
[0, 0, 0, 0, 1, 0, 0, 1]^\top,\quad [1, 0, 1, 0, 1, 1, 0, 0]^\top,\quad [1, 0, 0, 1, 0, 0, 0, 0]^\top.
\end{equation*}
One of their orthonormal choice is 
\begin{align}
\notag
\phi_1 &= \left[0, 0, 0, 0, \frac{1}{\sqrt{2}}, 0, 0, \frac{1}{\sqrt{2}}\right]^\top,\\
\label{ONB}
\phi_2 &= \left[\frac{1}{2\sqrt{3}}, 0, \frac{1}{\sqrt{3}}, -\frac{1}{2\sqrt{3}}, \frac{1}{2\sqrt{3}}, \frac{1}{\sqrt{3}}, 0, -\frac{1}{2\sqrt{3}}\right]^\top,\quad 
\phi_3 = \left[\frac{1}{\sqrt{2}}, 0, 0, \frac{1}{\sqrt{2}}, 0, 0, 0, 0\right]^\top.
\end{align}
\end{proposition}

\par
Next consider the perturbation of these eigenstructures.
Namely, consider eigenpairs of $\hat{W}_{s,t}(k)$ with sufficiently small $k$.
By the perturbation theory of eigenvalues of matrices, each eigenvalue of $\hat{W}_{s,t}(k)$ depends holomorphically on $k$ except several exceptional points such as eigenvalues with non-trivial multiplicity.
Singularities of eigenpairs with respect to $k$ can occur if the algebraic and geometric multiplicity of eigenvalues are different.
Nevertheless, our example for $\hat{W}_{s,t}(0)$ shows that these multiplicities coincide, since all eigenvalues are semisimple, and hence one expects that eigenstructures of $\hat{W}_{s,t}(k)$, including eigenfunctions, can be characterized as perturbations of corresponding unperturbed objects.
\begin{lemma}
\label{lemma-Kato} (\cite{Kato1982}) Consider the formal series of a matrix function $T(\kappa)$ with $\kappa\in \mathbb{C}$ of the form 
\begin{equation*}
T(\kappa) =  T + T^{(1)}\kappa + T^{(2)}\kappa^2 + \cdots
\end{equation*}
such that $T$  has the following structure of eigenvalues.
Assume that $\lambda \in {\rm Spec}(T)$ is semisimple.
Then, for small $\kappa$, every eigenvalue of $T(\kappa)$ near $\lambda$ has the following form:
there exists an eigenvalue $\lambda_j^{(1)}\in \mathrm{Spec}(\tilde{T}^{(1)})$ and $\alpha_{jk}\in \mathbb{C}$ and $p_j \in \mathbb{Z}_{>0}$ such that
\begin{equation}
\label{ev-perturb}
\lambda(\kappa) = \lambda + \lambda_j^{(1)}\kappa + \alpha_{jk}\kappa^{1+p_j^{-1}} + \cdots,
\end{equation}
where the matrix $\tilde T^{(1)}$ is described by $\tilde T^{(1)}\equiv \tilde T^{(1)}(0) = \Pi T^{(1)}\Pi$ with 
\begin{equation*}
\tilde T^{(1)}(\kappa) \equiv \frac{1}{\kappa} (T(\kappa) - \lambda)\Pi(\kappa).
\end{equation*}
Here $\Pi(\kappa)$ is the total projection associated with all the eigenvalues of $T(\kappa)$ close to $\lambda$ (more precisely, the $\lambda$-group) and 
the unperturbed projection $\Pi\equiv \Pi(0)$ is the  eigenprojection of eigenvalue $\lambda$ along its complementary subspace.
\end{lemma}
\begin{remark}\label{remark-Kato}
Note that the eigenfunction of the $\lambda'$-group whose cardinarity is one depends holomorphically on $\kappa$ in general (\cite{Kato1982}).
If we fix a semi-simple eigenvalue $\lambda$ of $T$ and an eigenvalue $\tilde{\lambda}_j^{(1)}$ of $\tilde{T}^{(1)}$, concerning the $(\lambda+\kappa \tilde{\lambda}_j^{(1)})$- group instead of the  $\lambda$-group, we can thus conclude that all eigenfunctions associated with eigenvalues of the form (\ref{ev-perturb}) converge to linear combinations of eigenfunctions of $T$ associated with $\lambda$ as $\kappa \to 0$, once we know that all the eigenvalues $\{\lambda_j^{(1)}\}_j$ of $\tilde{T}^{(1)}$ are simple.
In other words, eigenstructure of $T(\kappa)$ associated with $\lambda(\kappa)$ can be characterized by that of $T$ associated with $\lambda$.
\end{remark}

Our strategy here is to compute eigenvalue $\lambda_j^{(1)}$ of $\tilde T^{(1)}$ in Remark \ref{remark-Kato} and to show its simpleness.
In the present case, $T = \hat{W}_{s,t}(0)$, $\lambda = 1$ and $j = 1,2,3$. 
To this end, we need to prepare an explicit expression for the eigenprojection $\Pi$ of the eigenvalue $1$ of $\hat{W}_{s,t}(0)$. 
Since $\Pi$ is an eigenprojection, it holds that 
$\Pi^2=\Pi$ and $\Pi\hat{W}_{s,t}=\hat{W}_{s,t}\Pi$. In addition, we will show that $\Pi=\Pi^*$, that is, $\Pi$ is an orthogonal projection in the present case.
\begin{lemma}\label{lem:orthogonal}
Let the centered generalized eigenspace of $\hat{W}_{s,t}$ be defined by 
\[ 
\mathcal{H}_c:=\{\psi\in \ell^2(\{s,\dots,t\};\mathbb{C}^4) \;|\; (\hat{W}_{s,t}-\lambda)^m\psi=0 \mathrm{\;for\; some\;} |\lambda|=1,\;m\geq 1\}
\]
Then we have the following properties of $\mathcal{H}_c$. 
\begin{enumerate}
\item The centered generalized eigenspace is expressed by 
\begin{align*}
\mathcal{H}_c
&= \bigoplus_{|\lambda|=1}\ker(\lambda-\hat{W}_{s,t}) \\
&= \mathrm{span}\{ \chi_{s,t}^*\psi \;|\; \mathrm{supp}(\psi)\subset \{s,\dots,t\},\;\psi \mathrm{\;is\;an\;eigenvector\;of\;} \hat{W}_{-\infty,\infty} \}
\end{align*} 
\item The complementary invariant subspace of $\mathcal{H}_c$ can be expressed by $\mathcal{H}_c^\perp$. 
\end{enumerate}
\end{lemma}
\begin{proof}
This is a direct consequence of  Lemma 3.3, (3.12) and Lemma 3.4 in \cite{HS}. 
\end{proof}
Therefore Lemma~\ref{lem:orthogonal} implies that the eigenprojection $P$ of the eigenvalue $1$ is computed by using Proposition~\ref{prop:eigenvec} as follows.  
\begin{align*}
\Pi &= \phi_1\phi_1^*+\phi_2\phi_2^*+\phi_3\phi_3^* \\
%
&=
\frac{1}{12}\begin{bmatrix}
7 & 0 & 2 & 5 & 1 & 2 & 0 & -1\\
0 & 0 & 0 & 0 & 0 & 0 & 0 & 0 \\
2 & 0 & 4 & -2 & 2 & 4 & 0 & -2\\
5 & 0 & -2 & 7 & -1 & -2 & 0 & 1\\
1 & 0 & 2 & -1 & 7 & 2 & 0 & 5\\
2 & 0 & 4 & -2 & 2 & 4 & 0 & -2\\
0 & 0 & 0 & 0 & 0 & 0 & 0 & 0 \\
-1 & 0 & -2 & 1 & 5 & -2 & 0 & 7\\
\end{bmatrix},
\end{align*}
which is an orthogonal projection, in which case $\Pi^2 = \Pi$ and $\Pi^\ast = \Pi$ hold.

On the other hand,
\begin{equation*}
T^{(1)} = \begin{bmatrix}
t^{(1)} & O \\
O & t^{(1)}
\end{bmatrix},\quad 
t^{(1)}= -\frac{i}{2} \begin{bmatrix}
1 & 0 & 0 & -1 \\
1 & 0 & 0 & 1 \\
1 & 0 & 0 & 1 \\
1 & 0 & 0 & -1 \\
\end{bmatrix}
\end{equation*}
and hence
\begin{align*}
\Pi T^{(1)}\Pi &= \frac{-i}{6}
\begin{bmatrix}
1 & 0 & 1 & 0 & 1 & 1 & 0 & 0 \\
0 & 0 & 0 & 0 & 0 & 0 & 0 & 0 \\
1 & 0 & 0 & 1 & 1 & 0 & 0 & 1 \\
0 & 0 & 1 & -1 & 0 & 1 & 0 & -1 \\
1 & 0 & 1 & 0 & 1 & 1 & 0 & 0 \\
1 & 0 & 0 & 1 & 1 & 0 & 0 & 1 \\
0 & 0 & 0 & 0 & 0 & 0 & 0 & 0 \\
0 & 0 & 1 & -1 & 0 & 1 & 0 & -1 \\
\end{bmatrix}
\end{align*}

Direct calculations yield the following result.
\begin{proposition}\label{prop:PTP}
The matrix $\Pi T^{(1)}\Pi$ is a skew-Hermitian matrix. 
Eigenvalues of the matrix $\Pi T^{(1)}\Pi$ on the eigenspace $\mathrm{R}(\Pi)$ are $0$, $i / \sqrt{3}$ and $-i/\sqrt{3}$. In particular, all these eigenvalues are simple.
The corresponding eigenprojections are $\Pi^{(1)}_1:=v_1v_1^{^*}$, $\Pi^{(1)}_2:=v_2v_2^*$ and $\Pi^{(1)}_3:=v_3v_3^*$, respectively. 
Here 
\begin{align*} v_1&=[1/2, 0,0,1/2,-1/2,0,0,-1/2]^\top,\\
v_2&=\frac{1}{\sqrt{2}\;(3-\sqrt{3})}[2-\sqrt{3},0,1-\sqrt{3},1,2-\sqrt{3},1-\sqrt{3},0,1]^\top,\\
v_3&=\frac{1}{\sqrt{2}\;(3+\sqrt{3})}[2+\sqrt{3},0,1+\sqrt{3},1,2+\sqrt{3},1+\sqrt{3},0,1]^\top.
\end{align*}
\end{proposition}
\begin{remark}\noindent
\begin{enumerate}
\item $\Pi=\Pi_1^{(1)}+\Pi_2^{(1)}+\Pi_3^{(1)}$ and $\Pi_i^{(1)}\Pi_j^{(1)}=\delta_{ij}\Pi_i^{(1)}$ holds. 
\item The eigenvalue on $\mathrm{R}(1-\Pi)$ is also $0$. The eigenvector of  eigenvalue $\epsilon$ for the matrix $\Pi(T^{(1)}-\epsilon I)\Pi$ corresponds to that of $\mathrm{R}(\Pi)$. Taking $\epsilon\to 0$ to this eigenvector, we obtain the eigenvector $v_1$. See \cite{Kato1982} for more detail. 
\end{enumerate}
\end{remark}
As a consequence, combining Lemma~\ref{lemma-Kato} with Proposition~\ref{prop:PTP}, we obtain the following proposition, which characterizes the eigenstructure of $\hat{W}_{s,t}(k)$ for sufficiently small $k$.
\begin{proposition}\label{prop:eigensystem}
Eigenvalues of the matrix $\hat{W}_{s,t}(k)$ far from $0$ are
\begin{equation*}
1 + o(k),\quad 1 \pm \frac{i}{\sqrt{3}}k + o(k),\quad -\frac{1}{2} + O(k),\quad \frac{-1\pm i}{4} + O(k)
\end{equation*}
as $k\to 0$.
The correspoinding eigenprojections for the three eigenvalues
\[\left\{1 + o(k),\, 1 + \frac{i}{\sqrt{3}}k + o(k),\;1 - \frac{i}{\sqrt{3}}k + o(k) \right\}\] converge to $v_1v_1^*$, $v_2v_2^*$ and $v_3v_3^*$ as $k\to 0$, respectively, where
\begin{align*}
    v_1 &= \frac{1}{2}\;[1, 0,0,1,-1,0,0,-1]^\top, \\
    v_2 &= \frac{1}{\sqrt{2}\;(3-\sqrt{3})}\;[2-\sqrt{3},0,1-\sqrt{3},1,2-\sqrt{3},1-\sqrt{3},0,1]^\top, \\
    v_3 &= \frac{1}{\sqrt{2}\;(3+\sqrt{3})}\;[2+\sqrt{3},0,1+\sqrt{3},1,2+\sqrt{3},1+\sqrt{3},0,1]^\top.
\end{align*}
\end{proposition}